\documentclass{llncs}
\usepackage{graphics}
\usepackage{amsmath,amsfonts,amssymb,latexsym,stmaryrd}
\usepackage[PostScript=dvips]{diagrams}
\usepackage{epsfig} 
\usepackage{pstricks}

\newcommand{\bit}{\begin{itemize}}
\newcommand{\eit}{\end{itemize}\par\noindent}
\newcommand{\ben}{\begin{enumerate}}
\newcommand{\een}{\end{enumerate}\par\noindent}
\newcommand{\beq}{\begin{equation}}
\newcommand{\eeq}{\end{equation}\par\noindent}
\newcommand{\beqa}{\begin{eqnarray*}}
\newcommand{\eeqa}{\end{eqnarray*}\par\noindent}
\newcommand{\beqn}{\begin{eqnarray}}
\newcommand{\eeqn}{\end{eqnarray}\par\noindent}

\newarrow{Eq}=====
\newarrow{Dash}....{>}

\def\HH{{\cal H}}
\newcommand{\GG}{\mathcal{G}}
\newcommand{\CC}{\mathcal{C}}
\newcommand{\II}{\mathrm{I}}
\newcommand{\MM}{\mathcal{M}}
\newcommand{\VV}{\mathcal{V}}
\newcommand{\dd}{\llcorner}
\newcommand{\sdot}{\bullet}
\newcommand{\ddd}{\lrcorner}
\newcommand{\uu}{\ulcorner}
\newcommand{\uuu}{\urcorner}
\newcommand{\Nat}{\mathbb{N}}
\newcommand{\llpar}{\bindnasrepma}
\newcommand{\Tr}{\mathsf{Tr}}
\newcommand{\Mon}{\mathsf{M}}
\newcommand{\SMon}{\mathsf{SM}}
\newcommand{\CCl}{\mathsf{CC}}
\newcommand{\SCC}{\mathsf{SCC}}
\newcommand{\FM}{F_{\Mon}(\CC )}
\newcommand{\FSM}{F_{\SMon}(\CC )}
\newcommand{\FTr}{F_{\Tr}(\CC )}
\newcommand{\FCC}{F_{\CCl}(\CC )}
\newcommand{\FSCC}{F_{\SCC}(\CC )}
\newcommand{\FV}{F_{\VV}(\CC )}
\newcommand{\Grph}{\mathbf{Graph}}
\newcommand{\Loop}[1]{\mathcal{L}[#1]}
\newcommand{\InvSet}{\mathbf{InvSet}}
\newcommand{\InvCat}{\mathbf{InvCat}}
\newcommand{\InvCMon}{\mathbf{InvCMon}}
\newcommand{\ie}{\textit{i.e.}~}
\newcommand{\Ppi}{P_{\pi}}
\newcommand{\Ppio}{P_{\pi}^0}
\newcommand{\ppi}{p_{\pi}}
\newcommand{\LL}{\mathcal{L}}
\newcommand{\Ob}[1]{\mathsf{Ob}\; #1}
\newcommand{\Mor}[1]{\mathsf{Mor}\, #1}
\newcommand{\munion}{\uplus}
\newcommand{\mlb}{\mathbf{\{} \!\! |}
\newcommand{\mrb}{| \!\! \mathbf{\}}}
\newcommand{\ptensor}{\otimes}
\newcommand{\sgn}{\mathsf{sgn}}
\newcommand{\isoarrow}{\stackrel{\cong}{\rightarrow}}
\newcommand{\id}{\mathsf{id}}
\newcommand{\dunion}{\sqcup}
\newcommand{\EX}[2]{\mathrm{EX}(#1 ,#2 )}

\title{Abstract Scalars, Loops, and Free Traced and Strongly Compact Closed Categories}

\author{Samson Abramsky}

\institute{Oxford University Computing Laboratory\\
Wolfson Building, Parks Road, Oxford OX1 3QD, U.K.\\
\texttt{http://web.comlab.ox.ac.uk/oucl/work/samson.abramsky/}}

\begin{document}

\maketitle

\begin{abstract}
We study structures which have arisen in recent work by the present author and Bob Coecke on a categorical axiomatics for Quantum Mechanics; in particular, the notion of \emph{strongly compact closed category}. We explain how these structures support a notion of \emph{scalar} which allows quantitative aspects of physical theory to be expressed, and how the notion of strong compact closure emerges as a significant refinement of the more classical notion of compact closed category.

We then proceed to an extended discussion of free constructions for a sequence of progressively more complex kinds of structured category, culminating in the strongly compact closed case. The simple geometric and combinatorial ideas underlying these constructions are emphasized. We also discuss variations where a prescribed monoid of scalars can be `glued in' to the free construction.
\end{abstract}

\section{Introduction}
In this preliminary section, we will discuss the background and motivation for the technical results in the main body of the paper, in a fairly wide-ranging fashion. The technical material itself should be essentially self-contained, from the level of a basic familiarity with monoidal categories (for which see e.g. \cite{Mac}).

\subsection{Background}
In recent work \cite{AC1,AC2}, the present author and Bob Coecke have developed a categorical axiomatics for Quantum Mechanics, as a foundation for high-level approaches to quantum informatics: type systems, logics, and languages for quantum programming and quantum protocol specification.
The central notion in our axiomatic framework is that of \emph{strongly compact closed category}. It turns out that this rather simple and elegant structure suffices to capture most of the key notions for quantum informatics: compound systems, unitary operations, projectors, preparations of entangled states, Dirac bra-ket notation, traces, scalars, the Born rule. This axiomatic framework admits a range of models, including of course the Hilbert space formulation of quantum mechanics.

Additional evidence for the scope of the framework is provided by recent work of Selinger \cite{Sel05}. He shows that the framework of completely positive maps acting on generalized states represented by density operators, used in his previous work on the semantics of quantum programming languages \cite{Sel04}, fits perfectly into the framework of strongly compact closed categories.\footnote{Selinger prefers to use the term `dagger compact closed category', since the notion of adjoint which is formalized by the dagger operation $()^{\dagger}$ is a separate structure which is meaningful in a more general setting.}
He also showed that a simple construction (independently found and studied in some depth by Coecke \cite{Coe05}), which can be carried out completely generally at the level of strongly compact closed categories, corresponds to passing to the category of completely positive maps (and specializes exactly to this in the case of Hilbert spaces).

\subsection{Multiplicatives and Additives}
We briefly mention a wider context for these ideas. To capture the branching structure of measurements, and the flow of (classical) information from the result of a measurement to the future evolution of the quantum system, an additional \emph{additive} level of structure is required, based on a functor $\oplus$, as well as the \emph{multiplicative} level of the compact closed structure based around the tensor product (monoidal structure) $\otimes$. This delineation of additive and multiplicative levels of Quantum Mechanics is one of the conceptually interesting outcomes of our categorical axiomatics. (The terminology is based on that of Linear Logic \cite{Gir89} --- of which our structures can be seen as `collapsed models'). In terms of ordinary algebra, the multiplicative level corresponds to the multilinear-algebraic aspect of Quantum Mechanics, and the additive level to the linear-algebraic. But this distinction is usually lost in the sea of matrices; in particular, it is a real surprise how much can be done purely with the multiplicative structure.

It should be mentioned that we fully expect an \emph{exponential} level to become important, in the passage to the multi-particle, infinite dimensional, relativistic, and eventually field-theoretic levels of quantum theory.

We shall not discuss the additive level further in this paper. For most purposes, the additive structure can be regarded as freely generated, subject to arithmetic requirements on the scalars (see \cite{AC1}).

\subsection{Explicit constructions of free structured categories}
Our main aim in the present paper is to give explicit characterizations of free constructions for various kinds of categories-with-structure, most notably, for traced symmetric monoidal and strongly compact closed categories. We aim to give a synthetic account, including some basic cases which are well known from the existing literature \cite{Mac,KL80}. We will progressively build up structure through the following levels:
\begin{center}
\begin{tabular}{ll}
(1) $\quad$ & Monoidal Categories \\
(2) & Symmetric Monoidal Categories \\
(3) & Traced Symmetric Monoidal Categories \\
(4) & Compact Closed Categories \\
(5) & Strongly Compact Closed Categories \\
(6) & Strongly Compact Closed Categories with prescribed scalars \\
\end{tabular}
\end{center}
Of these, those cases which have not, to the best of our knowledge,, appeared previously are (3), (5) and (6). But in any event, we hope that our account will serve as a clear, accessible and useful reference.

Our constructions also serve to \emph{factor} the Kelly-Laplaza construction \cite{KL80} of the free compact closed category through the $\GG$ or $\mathrm{Int}$ construction \cite{JSV,Abr96} of the compact closed category freely generated by a traced symmetric monoidal category, which is a central part of (the mathematically civilised version of) the so-called `Geometry of Interaction' \cite{Gir89,AHS}.

It should be emphasized that constructions (1)--(4) are free over \emph{categories}, (5) over categories with involutions, and (6) over a comma category of categories with involution with a specified evaluation of scalars. We note that Dusko Pavlovic has give a free construction of traced categories over \emph{monoidal categories} \cite{Pav}. His construction is elegant, but  abstract and less combinatorial/geometric than ours: perhaps necessarily so, since in our situation the monoidal structure, which itself has some spatial content, is added freely.
Another reference is by Katis, Sabadini and Walters \cite{KSW}. They construct a free `feedback category', which is a trace minus the Yanking axiom --- which is very important for the dynamics of the trace --- over a monoidal category, and then formally quotient it to get a traced category.
A treatment in the same style as the present paper of free traced, compact closed and strongly compact closed categories over a monoidal category remains a topic for future investigation.

Furthermore, we will work entirely with the \emph{strict} versions of the categories-with-structure we will study. Since in each case, every such category is monoidally equivalent to a strict one, this does not really lose any generality; while by greatly simplifying the description of the free constructions, it makes their essential content, especially the geometry that begins to emerge as we add traces and compact closure (paths and loops), much more apparent.

\subsection{Diagrammatics} 
Our free constructions have immediate diagrammatic interpretations, which make their geometric content quite clear and vivid. Diagrammatic notation for tensor categories has been extensively developed with a view to applications in categorical formulations of topological invariants, quantum groups,  and topological quantum field theories \cite{Kas}. Within the purely categorical literature, a forerunner of these developments is the early work of Kelly on coherence \cite{Kel1,Kel2}; while the are also several precursors in the non-categorical literature, notably Penrose's diagrammatic notation for abstract tensors \cite{Pen}.

Diagrammatic notation has played an important role in our own work with Coecke on applying our categorical axiomatics to quantum informatics, e.g. to quantum protocols \cite{AC1}. For example, the essence of the verification of the teleportation protocol is the  diagrammatic equality shown in Figure~1.
For details, see \cite{AC1,AC2}.
\begin{figure}
\caption[]{Diagrammatic proof of teleportation}
\begin{center}
\psset{unit=1in,cornersize=absolute}%
\begin{pspicture}(-0.070175,-1.824561)(3.929825,0.701754)
\psset{linewidth=2pt}%
\pscustom[fillcolor=lightgray,fillstyle=solid,linecolor=blue]{%
\psline(0,-0.421053)(0.842105,-0.421053)
(0.421053,0)
(0,-0.421053)
}%
\pscustom[fillcolor=yellow,fillstyle=solid,linecolor=blue]{%
\psline(0.561404,-0.982456)(1.403509,-0.982456)
(0.982456,-1.403509)
(0.561404,-0.982456)
}%
\psframe[fillstyle=solid,fillcolor=lightgray,linecolor=blue](-0.003884,-0.775814)(0.284586,-0.487344)
\rput(0.140351,-0.631579){{ $\beta_i$}}
\psset{linewidth=1pt}%
\psframe[linestyle=dashed](-0.07406,-0.845989)(0.916165,0.07406)
\psframe[fillstyle=solid,fillcolor=green,linecolor=blue](1.118923,0.136467)(1.407393,0.424937)
\rput(1.263158,0.280702){{ $\beta_{i}^{-1}$}}
\psset{linewidth=2pt}%
\psline(0.140351,-1.403509)(0.140351,-0.77193)
\psline(0.140351,-0.491228)(0.140351,-0.350877)
(0.701754,-0.350877)
(0.701754,-1.052632)
(1.263158,-1.052632)
(1.263158,0.140351)
\psline[arrowsize=0.1in 0,arrowlength=1.25,arrowinset=0]{->}(1.263158,0.421053)(1.263158,0.701754)
\rput(1.964912,-0.421053){{\Huge $=$}}
\psset{linewidth=1pt}%
\psframe[fillstyle=solid,fillcolor=lightgray,linecolor=blue](2.522432,-0.775814)(2.810902,-0.487344)
\rput(2.666667,-0.631579){{ $\beta_i$}}
\psframe[fillstyle=solid,fillcolor=green,linecolor=blue](2.522432,0.136467)(2.810902,0.424937)
\rput(2.666667,0.280702){{ $\beta_{i}^{-1}$}}
\psset{linewidth=2pt}%
\psline(2.666667,-1.403509)(2.666667,-0.77193)
\psline(2.666667,-0.491228)(2.666667,0.140351)
\psline[arrowsize=0.1in 0,arrowlength=1.25,arrowinset=0]{->}(2.666667,0.421053)(2.666667,0.701754)
\rput(3.368421,-0.421053){{\Huge $=$}}
\psline[arrowsize=0.1in 0,arrowlength=1.25,arrowinset=0]{->}(3.929825,-1.403509)(3.929825,0.701754)
\end{pspicture}%
\end{center}
\end{figure}
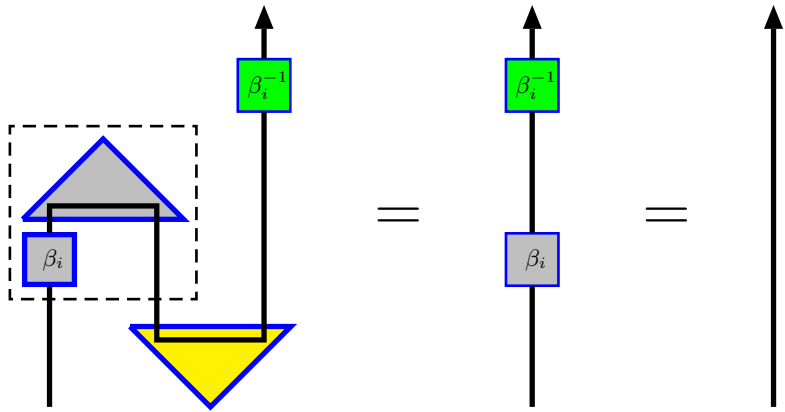

\subsection{Categorical Quantum Logic}
The diagrammatics of our constructions leads in turn to the idea of a \emph{logical} formulation, in which the diagrammatic representation of a morphism in the free category is thought of as a \emph{proof-net}, in the same general sense as in Linear Logic \cite{Gir87}.

More precisely, morphisms in the free category will correspond to proof nets in normal form, and the definition of composition in the category gives a direct construction for  normalizing a cut between two such proof nets. One advantage of the logical formulation is that we get an explicit syntactic description of these objects, and we can decompose the normalization process into cut-reduction steps, so that the computation of the normal form can be captured by a rewriting system.
This provides an explicit computational basis for deciding equality of proofs, which corresponds in the categorical context to verifying the \emph{commutativity of a diagram}.

In the categorical approach to quantum informatics \cite{AC1}, verifying the correctness of various quantum protocols is formulated as showing the commutativity of certain diagrams; so a computational theory of the above kind is directly applicable to such verifications.

In a joint paper with Ross Duncan \cite{AD05}, we have developed a system of Categorical Quantum Logic along these lines, incorporating additive as well as multiplicative features. This kind of logic, and its connection with Quantum Mechanics, is very different to the traditional notion of `Quantum Logic' \cite{BvN}.
Duncan is continuing to develop this approach in his forthcoming thesis.

\subsection{Overview}
The further structure of the paper is as follows. In Section~2 we explore the abstract notion of scalar which exists in any monoidal category. As we will see, scalars play an important role in determining the structure of free traced and strongly compact closed categories, as they correspond to the values of \emph{loops}. In Section~3, we review the notions of compact closed and strongly compact closed categories. The need for the notion of strong compact closure, to capture the structure of the complex spaces arising in Quantum Mechanics, is explained. In Section~4, we turn to the free constructions themselves.

\paragraph{Notation}
We set up some notation which will be useful.
We define $[n] := \{ 1, \ldots , n \}$ for $n \in \Nat$.
We write $S(n)$ for the symmetric group on $[n]$.
If $\pi \in S(n)$ and $\sigma \in S(m)$, we define
$\pi \ptensor \sigma \in S(n+m)$ by
\[ 
\pi \ptensor \sigma (i)  = \begin{cases}\pi (i), & 1 \leq i \leq n \\
\sigma (i-n) + n, & n+1 \leq i \leq n+m .
\end{cases}
\]
Given $\lambda : [n] \rightarrow X$, $\mu : [m] \rightarrow X$, we define $[\lambda , \mu ] : [n+m] \rightarrow X$ by
\[ [\lambda , \mu ] (i)  = \begin{cases} \lambda (i), & 1 \leq i \leq n \\
\mu (i-n), & n+1 \leq i \leq n+m .
\end{cases}
\]
We write $\MM (X)$ for the free commutative monoid generated by a set $X$. Concretely, these are the finite multisets over $X$, with the addition given by multiset union, which we write as $S \munion T$. 

\section{Scalars in monoidal categories}
The concept of a \emph{scalar} as a basis for quantitative measurements is fundamental in Physics. In particular, in Quantum Mechanics complex numbers $\alpha$ play the role of \emph{probability amplitudes}, with corresponding probabilities $\alpha \bar{\alpha} = | \alpha |^2$.

A key step in the development of the categorical axiomatics for Quantum Mechanics in \cite{AC1} was the recognition that the notion of scalar is meaningful in great generality --- in fact, in any monoidal (not necessarily symmetric) category.\footnote{Susbsequently, I became aware through Martin Hyland of the mathematical literature on Tannakian categories \cite{SR,Del}, stemming ultimately from Grothendiek. Tannakian categories embody much stronger assumptions than ours, in particular that the categories are abelian as well as compact closed, although the idea of strong compact closure is absent. But they certainly exhibit a consonant development of a large part of multilinear algebra in an abstract setting.}

Let $(\CC , \otimes , \II )$ be a strict monoidal category .
We define a \emph{scalar} in $\CC$ to be a morphism $s : \II \rightarrow \II$, \ie an endomorphism of the tensor unit.

\begin{example}
In $\mathbf{FdVec}_{\mathbb{K}}$, the category of finite-dimensional vector spaces over a field $\mathbb{K}$, linear maps $\mathbb{K} \to \mathbb{K}$ are uniquely
determined by the image of $1$, and hence correspond biuniquely to
elements of $\mathbb{K}\,$; composition corresponds to multiplication of 
scalars. In $\mathbf{Rel}$, there are just two scalars, corresponding
to the Boolean values $0$, $1$.
\end{example}

\noindent The (multiplicative) monoid of scalars is then just the endomorphism monoid $\CC (\II , \II )$.
The first key point is the elementary but beautiful observation by Kelly and Laplaza \cite{KL80} that this monoid is always commutative.
\begin{lemma}
\label{scprop}
$\CC (\II , \II )$ is a commutative monoid
\end{lemma}
\begin{proof}
\begin{diagram}
\II & \rEq & \II \otimes \II &  \rEq & \II \otimes \II & \rEq & \II \\
\uTo^{s} & & \uTo^{s \otimes 1} & & \dTo_{1 \otimes t} & & \dTo_{t} \\
\II & \rEq & \II \otimes \II & \rTo^{s \otimes t} & \II \otimes \II & \rEq & \II \\
\dTo^{t} & & \dTo^{1 \otimes t} & & \uTo_{s \otimes 1} & & \uTo_{s} \\
\II & \rEq & \II \otimes \II &  \rEq & \II \otimes \II & \rEq & \II \\
\end{diagram}
\end{proof}

\noindent We remark that in the non-strict case, where we have unit isomorphisms 
\[ \lambda_A : \II \otimes A \rightarrow A \qquad
\rho_A : A \otimes \II \rightarrow A \]
the proof makes essential use of the coherence axiom $\lambda_{\II} = \rho_{\II}$.

The second point is that a good notion of \emph{scalar multiplication} exists at this level of generality.
That is, each scalar
$s:\II\to\II$ induces a natural transformation
\[ \begin{diagram}
s_A :A & \rTo^{\simeq} & \II \otimes \!A & \rTo^{s \otimes 1_A} & \II
\otimes\! A & \rTo^{\!\!\simeq\ } & A\,.\ \ \ \ \
\end{diagram}
\]
with the naturality square
\[ \begin{diagram}
A & \rTo^{s_A} & A \\
\dTo^{f} & & \dTo_{f} \\
B & \rTo_{s_B} & B \\
\end{diagram}
\]
We write $s \sdot f$ for $f \circ s_A=s_B\circ f$.
Note that
\begin{align}
1 \sdot f & = f \label{sdotident} \\
s \sdot (t \sdot f) & =  (s \circ t) \sdot f \label{sdotact}\\
(s \sdot g)\circ(t \sdot f) & =  (s\circ t)\sdot(g\circ f) \label{sdotcomp}\\
(s \sdot f) \otimes (t \sdot g) & =  (s \circ t) \sdot (f \otimes g) \label{sdotten}
\end{align}
which exactly generalizes the multiplicative part of the usual properties of scalar multiplication.
Thus scalars act globally on the whole category.

\section{Strongly Compact Closed Categories}
A \em compact closed category \em is a symmetric monoidal
category in which to each object $A$ a \em dual \em
$A^*$, a \em unit \em $\eta_A:{\rm I}\to A^*\otimes A$ and a \em  
counit \em $\epsilon_A:A\otimes A^*\to {\rm I}$ are assigned in
such a way that the following `triangular identities' hold:
\[ 
\label{ccc1}
\begin{diagram}  
A &&\rTo^{1_A\otimes\eta_A} && A\otimes A^* \otimes
A && \rTo^{\epsilon_A\otimes
1_A}&& A\\ 
\end{diagram}
\quad = \quad 1_A 
\]
\[ 
\label{ccc2}
\begin{diagram}  
A^* &&\rTo^{\eta_A \otimes 1_A}&& A^* \otimes A \otimes
A^* && \rTo^{1_A \otimes \epsilon_A} && A^* \\ 
\end{diagram}
\quad = \quad 1_{A^*}
\]
Viewing monoidal categories as bicategories with a single 0-cell, this amounts to the axiom: 
\begin{center}
\fbox{Every object (1-cell) has an adjoint}
\end{center}
We can also view compact closed categories as *-autonomous categories \cite{Barr} for which ${\otimes} = {\llpar}$, and hence as `collapsed' models of Linear Logic \cite{Gir87}.

\subsection{Examples}

\begin{itemize}
\item $({\bf Rel},\times)$: Sets, relations,
and cartesian product.
Here $\eta_X\subseteq\{*\}\times(X\times X)$ and we have
\[
\eta_X=\epsilon_X^{c}=\{(*,(x,x))\mid x\in X\}\,.
\] 
\item $({\bf FdVec}_\mathbb{K},\otimes)$: Vector spaces over a field $\mathbb{K}$, linear maps,
and tensor product.
The unit and counit in $({\bf FdVec}_\mathbb{C},\otimes)$ are
\beqa 
&&\hspace{-6mm}\eta_V:\mathbb{C}\to V^*\otimes V::1\mapsto\sum_{i=1}^{i=n}\bar{e}_i\otimes e_i\\
&&\hspace{-6mm}\epsilon_V:V\otimes
V^*\to\mathbb{C}::e_j\otimes\bar{e}_i\mapsto\langle \bar{e}_{i} \mid e_{j}\rangle
\eeqa 
where $n$ is the dimension of $V$, $\{e_i\}_{i=1}^{i=n}$ is a basis for
$V$ and $\bar{e}_i$ is the linear functional in $V^*$ determined by
$\bar{e}_{j}( e_{i}) = \delta_{ij}$.
\end{itemize}

\subsection{Duality, Names and Conames}

For each morphism $f:A\to B$ in a compact closed category we can construct
a \em
dual
\em
$f^* : B^* \rightarrow A^*$:
\[ f^* \quad = \quad
\begin{diagram}
B^* & \rTo^{\eta_A \otimes 1} & A^* \otimes A \otimes B^* & \rTo^{1 \otimes f \otimes 1} &   A^* \otimes B \otimes B^* & \rTo^{1 \otimes \epsilon_B} & A^* \\
\end{diagram}
\]
a
\em name
\em 
\[ \uu f \uuu : \II \rightarrow A^* \otimes B \quad = \quad
\begin{diagram}
\II & \rTo^{\eta} & A^* \otimes A & \rTo^{1 \otimes f} & A^* \otimes B \\
\end{diagram}
\]
and a \em coname \em 
\[ \dd f \ddd : A \otimes B^* \rightarrow \II \quad = \quad
\begin{diagram}
A \otimes B^* & \rTo^{f \otimes 1} & B \otimes B^* & \rTo^{\epsilon} & \II \\
\end{diagram}
\]
The assignment $f\mapsto f^*$ extends $A\mapsto A^*$ into a
contravariant endofunctor with $A\simeq A^{**}$. In any
compact closed category, we have
\[ \CC (A \otimes B^\ast , \II ) \simeq \CC (A, B) \simeq \CC (\II ,
A^\ast \otimes B).
\]
For $R\in{\bf Rel}(X,Y)$ we have
\beqa
&&\hspace{-6mm}\uu R\uuu=\{(*,(x,y))\mid xRy,x\in X, y\in Y\}\\
&&\hspace{-6mm}\dd R\ddd=\{((x, y),*)\mid xRy,x\in X, y\in Y\}
\eeqa
and for $f\in {\bf FdVec}_\mathbb{K}(V,W)$ with matrix $(m_{ij})$ in bases
$\{e_i^V\}_{i=1}^{i=n}$ and $\{e_j^W\}_{j=1}^{j=m}$ of $V$ and $W$
respectively:
\beqa
&&\hspace{-6mm}\uu f\uuu:\mathbb{K}\to V^*\otimes
W::1\mapsto\sum_{i,j=1}^{\!i,j=n,m\!}m_{ij}\cdot
\bar{e}_i^V\otimes e_j^W\\ &&\hspace{-6mm}\dd f\ddd:V\otimes
W^*\to\mathbb{K}::e_i^V\otimes\bar{e}_j^W\mapsto m_{ij} .
\eeqa

\subsection{Why compact closure does not suffice}
In inner-product spaces we have the \emph{adjoint}:
\[ 
\begin{diagram}[loose,height=.8em,width=0pt]
& A && \rTo^{f} && B \\
\ & && \hLine && & \\
& A && \lTo^{f^{\dagger}} && B \\
\end{diagram} 
\qquad \qquad \langle f \phi \mid \psi \rangle_B = \langle \phi \mid f^{\dagger} \psi \rangle_A \]
This is \emph{not} the same as the dual --- the types are different.
In ``degenerate'' CCC's in which $A^* = A$, e.g. $\mathbf{Rel}$ or real inner-product spaces, we have $f^* = f^\dagger$.
In complex inner-product spaces such as Hilbert spaces, the inner product is \emph{sesquilinear}
\[ \langle \psi \mid \phi \rangle = \overline{\langle \phi \mid \psi \rangle} \]
and the isomorphism $A \simeq A^*$ is not linear, but \emph{conjugate linear}:
\[ \langle \lambda \sdot \phi \mid {-} \rangle \;\; = \;\; \bar{\lambda} \sdot \langle \phi \mid {-} \rangle 
\]
and hence does not live in the category $\mathbf{Hilb}$ at all!

\subsection{Solution: Strong Compact Closure}

We define the \emph{conjugate space} of a Hilbert space $\HH$: this has the same additive group of vectors as $\HH$, while the scalar multiplication and inner product are ``twisted'' by complex conjugation:
\[ \alpha \sdot_{\bar{\HH}} \phi := \bar{\alpha} \sdot_{\HH} \phi \qquad 
\langle \phi \mid \psi \rangle_{\bar{\HH}} := \langle \psi \mid \phi
\rangle_{\HH} \]
We can define $\HH^* = \bar{\HH}$, since $\HH$, $\bar{\HH}$ have the same orthornormal bases, and we can define the counit by
\[ \epsilon_{\HH} : \HH \otimes \bar{\HH} \rightarrow \mathbb{C} :: \phi \otimes \psi\mapsto \langle \psi \mid \phi\rangle_{\HH}
\]
which is indeed (bi)linear rather than sesquilinear!

The crucial observation is this: $()^*$ has a \emph{covariant} functorial extension $f \mapsto f_*$, which is essentially identity on morphisms; and then we can \emph{define}
\[ f^{\dagger} = (f^*)_* = (f_*)^* . \]

\subsection{Axiomatization of Strong Compact Closure}

In fact, there is a more concise and elegant axiomatization of strongly compact closed categories, which takes the adjoint as primitive \cite{AC2}.
It suffices to require the following structure on a (strict) symmetric monoidal category $(\CC , \otimes , \II , \tau )$:
\begin{itemize}
\item A strict monoidal involutive assignment
$A\mapsto A^*$ on objects.
\item An identity-on-objects, contravariant, strict monoidal, involutive functor
$f\mapsto f^\dagger$. 
\item For each object $A$ a unit
$\eta_A:\II\to A^*\otimes A$ with $\eta_{A^*}=\tau_{A^*\!,A}\circ\eta_A$
and such that either the diagram
\beq\label{sccc1}
\begin{diagram}
A&\rEq &A\otimes{\rm
I}&\rTo^{1_A\otimes\eta_A}&A\otimes(A^*\otimes
A)\\
\dTo^{1_A}&&&&\dEq \\
A&\lEq &{\rm I}\otimes
A&\lTo_{(\eta_A^\dagger\circ\tau_{A,A^*})\otimes 1_A}&(A\otimes
A^*)\otimes A
\end{diagram}
\eeq
or the diagram
\beq\label{sccc2}
\begin{diagram}
A&\rEq &{\rm I}\otimes
A&\rTo^{\eta_A\otimes 1_A}&(A^*\!\otimes A)\otimes
A&\rEq&A^*\!\otimes(A\otimes
A)\\
\dTo^{1_A}&&&&&&\dTo~{1_{A^*}\!\otimes\tau_{A,A}\!\!\!}\\
A&\lEq &{\rm I}\otimes
A&\lTo_{\eta_A^\dagger\otimes 1_A}&(A^*\!\otimes A)\otimes
A&\lEq&A^*\!\otimes(A\otimes A)
\end{diagram}
\eeq
commutes, where $\tau_{A,A}:A\otimes A\simeq A\otimes A$ is the twist
map.
\item Given such a functor $()^{\dagger}$, we  define an
isomorphism $\alpha$ to be \emph{unitary} if $\alpha^{-1} =
\alpha^{\dagger}$. We additionally require that the canonical natural isomorphism
for symmetry given as part of the symmetric
monoidal structure on $\CC$ is (componentwise) unitary in this sense.
\end{itemize}

While diagram (\ref{sccc1}) is the analogue to (\ref{ccc1}) with
$\eta_A^\dagger\circ\tau_{A,A^*}$ playing the role of the counit,
diagram
(\ref{sccc2}) expresses \em Yanking \em with respect to the canonical
trace of
the compact closed structure.\footnote{In fact,  we have used the `left trace' here rather than the more customary `right trace' which we shall use in our subsequent discussion of traced monoidal categories. In the symmetric context, the two are equivalent; we chose the left trace here because, given our other notational conventions, it requires less use of symmetries in stating the axiom.}   We only need one
commuting diagram as compared to (\ref{ccc1}) and (\ref{ccc2}) in
the
definition of compact closure,
since due
to the strictness assumption (i.e.~$A\mapsto A^*$ being involutive) we
were able
to replace the second diagram by
$\eta_{A^*}=\tau_{A^*\!,A}\circ\eta_A$.

\newpage
\subsubsection{Standard triangular identities diagrammatically}

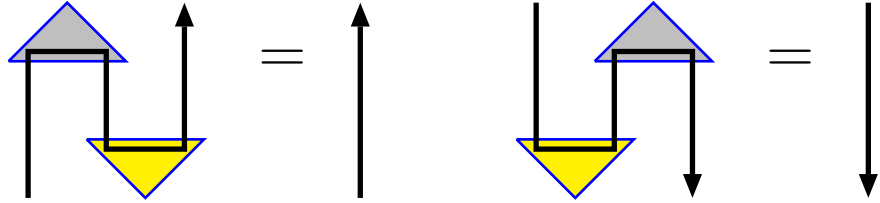
\begin{center}
\psset{unit=1in,cornersize=absolute}%
\begin{pspicture}(0,-1.329545)(4.5,0.306818)
\psset{linewidth=1pt}%
\pscustom[fillcolor=lightgray,fillstyle=solid,linecolor=blue]{%
\psline(0,-0.306818)(0.613636,-0.306818)
(0.306818,0)
(0,-0.306818)
}%
\pscustom[fillcolor=yellow,fillstyle=solid,linecolor=blue]{%
\psline(0.409091,-0.715909)(1.022727,-0.715909)
(0.715909,-1.022727)
(0.409091,-0.715909)
}%
\psset{linewidth=2pt}%
\psline[arrowsize=0.1in 0,arrowlength=1.25,arrowinset=0]{->}(0.102273,-1.022727)(0.102273,-0.255682)
(0.511364,-0.255682)
(0.511364,-0.767045)
(0.920455,-0.767045)
(0.920455,0)
\rput(1.431818,-0.306818){{\Huge $=$}}
\psset{linewidth=2pt}%
\psline[arrowsize=0.1in 0,arrowlength=1.25,arrowinset=0]{->}(1.840909,-1.022727)(1.840909,0)
\psset{linewidth=1pt}%
\pscustom[fillcolor=lightgray,fillstyle=solid,linecolor=blue]{%
\psline(3.068182,-0.306818)(3.681818,-0.306818)
(3.375,0)
(3.068182,-0.306818)
}%
\pscustom[fillcolor=yellow,fillstyle=solid,linecolor=blue]{%
\psline(2.659091,-0.715909)(3.272727,-0.715909)
(2.965909,-1.022727)
(2.659091,-0.715909)
}%
\psset{linewidth=2pt}%
\psline[arrowsize=0.1in 0,arrowlength=1.25,arrowinset=0]{->}(2.761364,0)(2.761364,-0.767045)
(3.170455,-0.767045)
(3.170455,-0.255682)
(3.579545,-0.255682)
(3.579545,-1.022727)
\rput(4.090909,-0.306818){{\Huge $=$}}
\psset{linewidth=2pt}%
\psline[arrowsize=0.1in 0,arrowlength=1.25,arrowinset=0]{->}(4.5,0)(4.5,-1.022727)
\end{pspicture}%
\end{center}

\[ (\epsilon_A \otimes 1_A ) \circ (1_A \otimes \eta_A ) = 1_A \qquad \qquad \qquad  \qquad
(1_{A^*} \otimes \epsilon_A ) \circ (\eta_A \otimes 1_{A^*} ) = 1_{A^*}
\]

\subsubsection{Yanking diagrammatically}
\begin{center}
\psset{unit=1in,cornersize=absolute}%
\begin{pspicture}(0,-1.181818)(2,0.272727)
\psset{linewidth=1pt}%
\pscustom[fillcolor=lightgray,fillstyle=solid,linecolor=blue]{%
\psline(0,-0.272727)(0.545455,-0.272727)
(0.272727,0)
(0,-0.272727)
}%
\pscustom[fillcolor=yellow,fillstyle=solid,linecolor=blue]{%
\psline(0,-0.636364)(0.545455,-0.636364)
(0.272727,-0.909091)
(0,-0.636364)
}%
\psset{linewidth=2pt}%
\psline[arrowsize=0.1in 0,arrowlength=1.25,arrowinset=0]{->}(0.727273,-1.181818)(0.727273,-0.636364)
(0.454545,-0.227273)
(0.090909,-0.227273)
(0.090909,-0.681818)
(0.454545,-0.681818)
(0.727273,-0.272727)
(0.727273,0.272727)
\rput(1.272727,-0.272727){{\Huge $=$}}
\psset{linewidth=2pt}%
\psline[arrowsize=0.1in 0,arrowlength=1.25,arrowinset=0]{->}(2,-1.181818)(2,0.272727)
\end{pspicture}%
\end{center}
\[ (\eta_A^{\dagger} \otimes 1_{A}) \circ (1_{A^*} \otimes \tau_{A, A}) \circ (\eta_A \otimes 1_{A}) = 1_A .
\]

\section{Free Constructions}
We will now give detailed descriptions of free constructions for a number of types of category-with-structure.
We shall consider the following cases:
\begin{center}
\begin{tabular}{ll}
(1) $\quad$ & Monoidal Categories \\
(2) & Symmetric Monoidal Categories \\
(3) & Traced Symmetric Monoidal Categories \\
(4) & Compact Closed Categories \\
(5) & Strongly Compact Closed Categories \\
(6) & Strongly Compact Closed Categories with prescribed scalars \\
\end{tabular}
\end{center}
For cases (1)--(4), we shall consider adjunctions of the form
\[ \begin{diagram}
\mathbf{Cat} && \pile{\rTo^{F_S}\\ \bot \\ \lTo_{U_S}} && \mathbf{S{-}Cat} \\
\end{diagram}
\]
where $S$ ranges over the various kinds of structure.
Specifically, we shall give explicit descriptions in each case of $F_S (\CC )$ for a category $\CC$. This explicit description --- not algebraically by generators and relations, but giving direct combinatorial definitions of the normal forms and how they compose, thus solving the \textbf{word problem} over these categories --- is the strongest form of coherence theorem available for notions such as compact closure and traces. In these cases, cyclic structures arise, violating the compatibility requirements for stronger forms of coherence developed in \cite{Kel1,Kel2}. This point is discussed in the concluding section of \cite{KL80}.

In case (5), we consider an adjunction
\[ \begin{diagram}
\mathbf{InvCat} && \pile{\rTo^{F_{\SCC}}\\ \bot \\ \lTo_{U_{\SCC}}} && \mathbf{SCC{-}Cat} \\
\end{diagram}
\]
where $\mathbf{InvCat}$ is the category of categories with a specified involution, (what Selinger calls `dagger categories' in \cite{Sel05}), and functors which preserve the involution.
Finally, in (6) we consider an adjunction with respect to a comma category, which allows us to describe the free strongly compact slosed category generated by a category $\CC$, together with a prescribed multiplicative monoid of scalars.

Our treatment will be incremental, reflecting the fact that in our sequence (1)--(6), each term arises by adding structure to the previous one.
Each form of structure is reflected conceptually by a new feature arising in the corresponding free construction:
\begin{center}
\begin{tabular}{lll}
$\mathsf{M} \qquad$ & monoidal & lists \\
$\mathsf{SM}$ & symmetric monoidal & permutations \\
$\mathsf{Tr}$ & traced symmetric monoidal $\quad$ & loops \\
$\mathsf{CC}$ & compact closed & polarities \\
$\mathsf{SCC}$ & strong compact closed & reversals \\
\end{tabular}
\end{center}

\noindent We will also begin to see a primitive graph-theoretic geometry of \emph{points}, \emph{lines} and \emph{paths} begin to  emerge as we progress through the levels of structure. There is in fact more substantial geometry lurking here than might be apparent: the elaboration of these connections must be left to future work.

Finally, we mention a recurring theme. To form a `pure' picture of each construction, it is useful to consider the case
$F_S (\mathbf{1})$ explicitly, where $\mathbf{1}$ is the category with (one object and) one morphism (\ie one generator, no relations).

\subsection{Monoidal Categories}
We begin with the simple case of monoidal categories.
The objects of $F_{\Mon} (\CC )$, the free monoidal category generated by the category $\CC$, are  lists of objects of $\CC$.
The (strict) monoidal structure is given by concatenation; the tensor unit $\II$ is the empty sequence.

\noindent Arrows:
\[ \begin{diagram}[loose,height=2em]
\raisebox{-4mm}{$A_1$} & \raisebox{-4mm}{$A_2$} & & \raisebox{-4mm}{$A_n$} \\
\bullet &\bullet & \cdots & \bullet \\
\dTo^{f_1} & \dTo^{f_2} & & \dTo_{f_n} \\
\bullet &\bullet & \cdots & \bullet \\
\raisebox{4mm}{$B_1$} & \raisebox{4mm}{$B_2$} & & \raisebox{4mm}{$B_n$} \\
\end{diagram}
\qquad\qquad f_i : A_i \rightarrow B_i
\]
An arrow from one list of objects to another is simply a list of arrows of $\CC$ of the appropriate types. Note that there can only be an arrow between lists of the same length.
Composition is performed pointwise in the obvious fashion.

\noindent Formally, we set $[n] := \{ 1, \ldots , n \}$, and define an object of $\FM$ to be a pair $(n, A)$, where $n \in \Nat$, and 
$A : [n] \rightarrow \Ob{\CC}$.
Tensor product of objects is defined by $(n, A) \otimes (m, B) = (n+m, [A, B])$. The tensor unit is $\II = (0, {!})$, where $!$ is the unique function from the empty set.

A morphism $\lambda : (n, A) \rightarrow (m, B)$ can only exist if $n = m$, and is specified by a map $\lambda : [n] \rightarrow \Mor{\CC}$, satisfying
\[ \lambda_i : A_i \longrightarrow B_i . \]

\noindent Arrows in $F_{\Mon}(\CC )$ are thus simply those expressible in the form
\[ f_1 \otimes \cdots \otimes f_k : A_1 \otimes \cdots \otimes A_k \longrightarrow B_1 \otimes \cdots \otimes B_k . \]
Unicity of the monoidal functor to a monoidal category $\mathcal{M}$ extending a given functor $F : \CC \rightarrow U_{\Mon} \mathcal{M}$ is then immediate.

\noindent Note that 
\[ F_{\Mon} (\mathbf{1}) = (\Nat , = , + , 0) . \]

\subsection{Symmetric Monoidal Categories}
The  objects of $\FSM$ are the same as in the monoidal case.

An arrow $(n, A) \longrightarrow (n, B)$ is given by $(\pi , \lambda )$, where $\pi \in S(n)$ is a permutation, and
$\lambda_i : A_i \rightarrow B_{\pi (i)}$, $1 \leq i \leq n$.
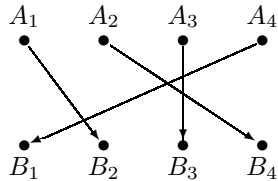
\begin{diagram}[height=2em,width=1.5em,abut]
\raisebox{-4mm}{$A_1$} & & \raisebox{-4mm}{$A_2$} & & \raisebox{-4mm}{$A_3$} & & \raisebox{-4mm}{$A_4$} \\
\bullet &&\bullet && \bullet && \bullet \\
& \rdTo & &  \rdTo(4,2) & \dTo & \ldTo(6,2) & \\
\bullet &&\bullet && \bullet && \bullet \\
\raisebox{4mm}{$B_1$} && \raisebox{4mm}{$B_2$} && \raisebox{4mm}{$B_3$} && \raisebox{4mm}{$B_4$} \\
\end{diagram}

\noindent Composition in $F_{\SMon}(\CC )$
is described as follows.
Form paths of length 2, and compose the arrows from $\CC$ labelling these paths.
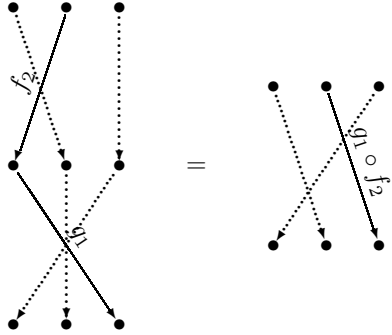
\[
\begin{diagram}[width=1em,abut]
\bullet & & \bullet & & \bullet \\
& \rdDash \ldTo^{f_2} & & & \dDash \\
\bullet & & \bullet & & \bullet \\
& \rdTo(4,2)^{g_1} & \dDash & \ldDash(4,2) \\
\bullet & & \bullet & & \bullet \\
\end{diagram}
\qquad = \qquad
\begin{diagram}[width=1em,abut]
\bullet & & \bullet & & \bullet \\
& \rdDash & & \rdTo^{g_1 \circ f_2}  \ldDash(4,2) & \\
\bullet & & \bullet & & \bullet \\
\end{diagram}
\]

\noindent Note that
$F_{\SMon}(\mathbf{1}) = \coprod_{n} S(n)$ (coproduct of categories). Thus the free monoidal category on the trivial generating category comprises (the disjoint union of) all the finite symmetric groups.\footnote{At this point, a possible step towards geometry presents itself. If we considered free braided monoidal categories, we would find a similar connection to the braid groups \cite{Kas}. However, we shall not pursue that here.}

\noindent Let $\mathcal{M}$ be a symmetric monoidal category, and consider a tensor product $A_1 \otimes \cdots \otimes A_n$.
Each element $\pi \in S(n)$ of the symmetric group $S(n)$ induces an isomorphism, which by abuse of notation we also write as $\pi$:
\[ \pi : A_1 \otimes \cdots \otimes A_n \isoarrow
A_{\pi (1)} \otimes \cdots \otimes A_{\pi (n)} . \]
Now note that under the above concrete description of $\FSM$, arrows 
\[ (\pi , \lambda ) : (n, A) \longrightarrow (n, B) \]  
can be written as 
\begin{equation}
\label{smonform}
 \pi^{-1} \circ \bigotimes_{i=1}^n f_i \; : \; \bigotimes_{i=1}^n A_i \; \longrightarrow   \; \bigotimes_{i=1}^n B_i . 
\end{equation}
Again, the freeness property follows directly.
The main observation to be made is that such arrows are closed under composition:
\begin{equation}
\label{smoncomp}
(\sigma^{-1} \circ \bigotimes_{i=1}^n g_i ) \circ (\pi^{-1} \circ \bigotimes_{i=1}^n f_i ) =
(\sigma \circ \pi )^{-1} \circ \bigotimes_{i=1}^n (g_{\pi (i)} \circ f_i )  
\end{equation}
and tensor product:
\begin{equation}
\label{smontens}
(\pi^{-1} \circ \bigotimes_{i} f_i ) \otimes (\sigma^{-1} \circ \bigotimes_{i} g_i ) \;\; = \;\; (\pi \otimes \sigma)^{-1} \circ \left( \bigotimes_{i} f_i \otimes \bigotimes_{i} g_i \right) , 
\end{equation}
where if $\pi \in S(n)$, $\sigma \in S(m)$, $\pi \otimes \sigma \in S(n+m)$ is the evident concatenation of the two permutations, as defined in the 
Introduction.

\noindent The above closed form expression for composition requires the `naturality square':
\[ \begin{diagram}
\bigotimes_i B_{\pi(i)} && \rTo^{\pi^{-1}} && \bigotimes_i B_i \\
\dTo^{\bigotimes_i g_{\pi(i)}} && && \dTo_{\bigotimes_i g_i} \\
\bigotimes_i C_{\sigma \circ \pi(i)} && \rTo_{\sigma \circ (\sigma \circ \pi )^{-1}} && \bigotimes_i C_{\sigma (i)} \\
\end{diagram}
\]

\subsection{Traced Symmetric Monoidal Categories}
We now come to a key case, that of traced symmetric monoidal categories. Much of the structure of strongly compact closed categories in fact appears already at the traced level. This is revealed rather clearly by our incremental development of the free constructions.

We begin by recalling the basic notions. let $(\CC , \otimes , \II , \tau)$ be a symmetric monoidal category. 
Here $\tau_{A,B} : A \otimes B \isoarrow B \otimes A$ is the symmetry or twist natural isomorphism.
A \emph{trace} on $\CC$ is a family of functions
\[ \Tr_{A,B}^U : \CC (A \otimes U, B \otimes U) \longrightarrow \CC (A, B) \]
for objects $A$, $B$, $U$ of $\CC$, satisfying the following axioms:
\begin{itemize}
\item {\bf Input Naturality}:
\[ \mathsf{Tr}^U_{A, B}(f) \circ g=\mathsf{Tr}^U_{A',B}(f \circ (g\otimes 1_U))\]
where 
$f: A\otimes U \to B \otimes U$, $g: A' \to A$,

\item {\bf Output Naturality}:
\[ g \circ \mathsf{Tr}^U_{A, B}(f)=\mathsf{Tr}^U_{A,B'}((g\otimes 1_U) \circ f)\]
where 
$f: A\otimes U \to B \otimes U$, $g:B \to B'$,
\item {\bf Feedback Dinaturality}:
\[ \mathsf{Tr}^U_{A,B}((1_B\otimes g) \circ f) = \mathsf{Tr}^{U'}_{A,B}(f \circ (1_A\otimes g)) \]
where $f: A\otimes U \to B \otimes U'$, $g: U' \to U$,
\item {\bf Vanishing (I,II)}:
\[ \mathsf{Tr}^I_{A,B}(f)=f \qquad  \mbox{and} \qquad \mathsf{Tr}^{U \otimes V}_{A,B}(g)=\mathsf{Tr}^U_{A,B}(\mathsf{Tr}^V_{A\otimes U,B \otimes U}(g)) \]
where
$f: A\otimes I \to B\otimes I$ and 
$g: A\otimes U \otimes V \to B\otimes U \otimes V$.
\item {\bf Superposing}:
\[g \otimes \mathsf{Tr}_{A,B}^U(f) = \mathsf{Tr}_{W\otimes A,Z\otimes B}^U(g \otimes f)\]
where $f:A\otimes U \to B \otimes U$ and $g:W \to Z$ .
\item {\bf Yanking}:
\[ \mathsf{Tr}^U_{U,U}(\tau_{U,U}) = 1_U . \]
\end{itemize}

\noindent Diagrammatically, we depict the trace as feedback:
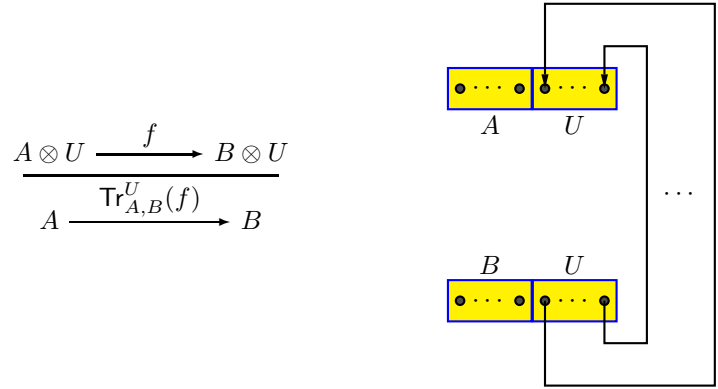
\[
\raisebox{1in}{\begin{diagram}[loose,height=.8em,width=0pt]
& A \otimes U  & \rTo^{f} & B \otimes U \\
\ & & \hLine & & \\
& A & \rTo^{\Tr_{A,B}^{U}(f)} & B \\
\end{diagram}}
\qquad \qquad \qquad 
\psset{unit=1in,cornersize=absolute}%
\begin{pspicture}(0,-1.555556)(1.4,0.444444)
\psframe[fillstyle=solid,fillcolor=yellow,linecolor=blue](-0.00246,-0.113571)(0.446904,0.113571)
\rput(0.222222,0){$\cdots$}
\psframe[fillstyle=solid,fillcolor=yellow,linecolor=blue](0.441985,-0.113571)(0.891349,0.113571)
\rput(0.666667,0){$\cdots$}
\uput{0.5ex}[d](0.222222,-0.111111){$A$}
\uput{0.5ex}[d](0.666667,-0.111111){$U$}
\newgray{fillval}{0.3}
\pscircle[fillstyle=solid,fillcolor=fillval](0.066667,0){0.027757}
\pscircle[fillstyle=solid,fillcolor=fillval](0.377778,0){0.027757}
\pscircle[fillstyle=solid,fillcolor=fillval](0.511111,0){0.027757}
\pscircle[fillstyle=solid,fillcolor=fillval](0.822222,0){0.027757}
\psframe[fillstyle=solid,fillcolor=yellow,linecolor=blue](-0.00246,-1.224682)(0.446904,-0.99754)
\rput(0.222222,-1.111111){$\cdots$}
\psframe[fillstyle=solid,fillcolor=yellow,linecolor=blue](0.441985,-1.224682)(0.891349,-0.99754)
\rput(0.666667,-1.111111){$\cdots$}
\uput{0.5ex}[u](0.222222,-1){$B$}
\uput{0.5ex}[u](0.666667,-1){$U$}
\pscircle[fillstyle=solid,fillcolor=fillval](0.066667,-1.111111){0.027757}
\pscircle[fillstyle=solid,fillcolor=fillval](0.377778,-1.111111){0.027757}
\pscircle[fillstyle=solid,fillcolor=fillval](0.511111,-1.111111){0.027757}
\pscircle[fillstyle=solid,fillcolor=fillval](0.822222,-1.111111){0.027757}
\psline[arrowsize=0.0375in 0,arrowlength=2,arrowinset=0]{->}(0.822222,-1.111111)(0.822222,-1.333333)
(1.044444,-1.333333)
(1.044444,0.222222)
(0.822222,0.222222)
(0.822222,0)
\psline[arrowsize=0.0375in 0,arrowlength=2,arrowinset=0]{->}(0.511111,-1.111111)(0.511111,-1.555556)
(1.4,-1.555556)
(1.4,0.444444)
(0.511111,0.444444)
(0.511111,0)
\rput(1.222222,-0.555556){$\cdots$}
\end{pspicture}%
\]
It corresponds to \emph{contracting indices} in traditional tensor calculus.

We now consider the free symmetric monoidal category generated by $\CC$, $F_{\SMon}(\CC )$, as described in the previous section. Recall that morphisms in $F_{\SMon}(\CC )$ can be written as
\[ \pi^{-1} \circ \bigotimes_{i=1}^n f_i \; : \; \bigotimes_{i=1}^n A_i \; \longrightarrow \;  \bigotimes_{i=1}^n B_i . \]
Our first observation is that \emph{this category is already canonically traced}. Understanding why this is so, and why $F_{\SMon}(\CC )$ is \emph{not} the free traced category, will lay bare the essential features of the free construction we are seeking.

Note firstly that, if there is an arrow $f : (n, A) \otimes (p, U) \rightarrow (m, B) \otimes (p, U)$ in $F_{\SMon}(\CC )$, then we must have $n +p = m+p$, and hence $n=m$.
Thus we can indeed hope to form an arrow $A \rightarrow B$ in $F_{\SMon}(\CC )$.
Now we consider the `geometry' arising from the permutation $\pi$, together with the diagrammatic feedback interpretation of the trace. We illustrate this with the following example.
\subsubsection{Example}
Consider the arrow $f = \pi^{-1} \circ \bigotimes_{i=1}^4 f_i$, where $\pi = (2, 4, 3, 1)$, and $f_i : A_i \rightarrow B_{\pi (i)}$. Suppose that $A_i = U_i = B_i$, $2 \leq i \leq 4$, and write $U = \bigotimes_{i=2}^4 U_i$. We wish to compute $\Tr_{A_1 , B_1}^U (f)$. The geometry is made clear by the following figure.
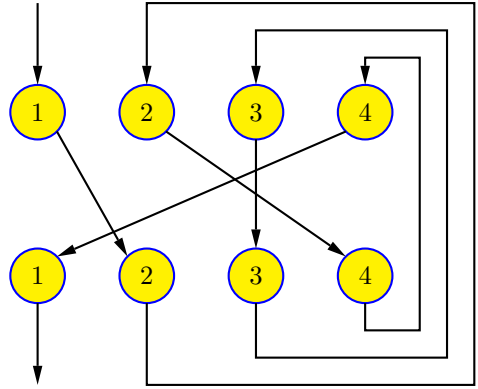
\begin{center}
\psset{unit=1in,cornersize=absolute}%
\begin{pspicture}(0,-1.428571)(2.428571,0.571429)
\pscircle[fillstyle=solid,fillcolor=yellow,linecolor=blue](0.142857,0){0.148392}
\rput(0.142857,0){$1$}
\pscircle[fillstyle=solid,fillcolor=yellow,linecolor=blue](0.714286,0){0.148392}
\rput(0.714286,0){$2$}
\pscircle[fillstyle=solid,fillcolor=yellow,linecolor=blue](1.285714,0){0.148392}
\rput(1.285714,0){$3$}
\pscircle[fillstyle=solid,fillcolor=yellow,linecolor=blue](1.857143,0){0.148392}
\rput(1.857143,0){$4$}
\pscircle[fillstyle=solid,fillcolor=yellow,linecolor=blue](0.142857,-0.857143){0.148392}
\rput(0.142857,-0.857143){$1$}
\pscircle[fillstyle=solid,fillcolor=yellow,linecolor=blue](0.714286,-0.857143){0.148392}
\rput(0.714286,-0.857143){$2$}
\pscircle[fillstyle=solid,fillcolor=yellow,linecolor=blue](1.285714,-0.857143){0.148392}
\rput(1.285714,-0.857143){$3$}
\pscircle[fillstyle=solid,fillcolor=yellow,linecolor=blue](1.857143,-0.857143){0.148392}
\rput(1.857143,-0.857143){$4$}
\psline[arrowsize=0.05in 0,arrowlength=2,arrowinset=0]{->}(0.243872,-0.101015)(0.61327,-0.756128)
\psline[arrowsize=0.05in 0,arrowlength=2,arrowinset=0]{->}(0.815301,-0.101015)(1.756128,-0.756128)
\psline[arrowsize=0.05in 0,arrowlength=2,arrowinset=0]{->}(1.285714,-0.142857)(1.285714,-0.714286)
\psline[arrowsize=0.05in 0,arrowlength=2,arrowinset=0]{->}(1.756128,-0.101015)(0.243872,-0.756128)
\psline[arrowsize=0.05in 0,arrowlength=2,arrowinset=0]{->}(1.857143,-1)(1.857143,-1.142857)
(2.142857,-1.142857)
(2.142857,0.285714)
(1.857143,0.285714)
(1.857143,0.142857)
\psline[arrowsize=0.05in 0,arrowlength=2,arrowinset=0]{->}(1.285714,-1)(1.285714,-1.285714)
(2.285714,-1.285714)
(2.285714,0.428571)
(1.285714,0.428571)
(1.285714,0.142857)
\psline[arrowsize=0.05in 0,arrowlength=2,arrowinset=0]{->}(0.714286,-1)(0.714286,-1.428571)
(2.428571,-1.428571)
(2.428571,0.571429)
(0.714286,0.571429)
(0.714286,0.142857)
\psline[arrowsize=0.05in 0,arrowlength=2,arrowinset=0]{<-}(0.142857,0.142857)(0.142857,0.571429)
\psline[arrowsize=0.05in 0,arrowlength=2,arrowinset=0]{->}(0.142857,-1)(0.142857,-1.428571)
\end{pspicture}%
\end{center}
We simply follow the path leading from $A_1$ to $B_1$: 
\[ 1 \rightarrow 2 \rightarrow 4 \rightarrow 1 \]
composing the arrows which label the arcs in the path: thus
\[ \Tr_{A_1 , B_1}^U (f) = f_4 \circ f_2 \circ f_1 \]
in this case. A similar procedure can always be followed for arrows in the form~(\ref{smonform}), which as we have seen is general for $\FSM$. (It is perhaps not immediately obvious that a path from an input will always emerge from the feedback zone into an output. See the following Proposition~\ref{permprop}). Moreover, this assignment does lead to a well-defined trace on $\FSM$.
However, this is \emph{not} the free traced structure generated by $\CC$.

To see why this construction does not give rise to the free interpretation of the trace, note that in our example, the node $U_3$ is involved in a cycle $U_3 \rightarrow U_3$, which does not appear in our expression for the trace of $f$. In fact, note that if we trace an endomorphism $f : A \rightarrow A$ out completely, \ie writing $f : \II \otimes A \rightarrow \II \otimes A$ we form $\Tr_{\II , \II}^{A}(f) : \II \rightarrow \II$, then we get a \emph{scalar}. Indeed, the importance of scalars in our context is exactly that they give the values of loops. Now in $F_{\SMon}(\CC )$, the tensor unit is the empty list, and there is only one scalar --- the identity. It is exactly this collapsed interpretation of the scalars which prevents the trace we have just (implicitly) defined on $F_{\SMon}(\CC )$ from giving the free traced category on $\CC$.

We now turn to a more formal account, culminating in the construction of $F_{\Tr}(\CC )$.

\subsubsection{Geometry of permutations}
We begin with a more detailed analysis of permutations $\pi \in S(n+m)$, with the decomposition $n +m$ reflecting our distinction between the visible (input-output) part of the type, and the hidden (feedback) part, arising from the application of the trace.

We define an \emph{$n$-path} (or if $n$ is understood, an \emph{input-output path}) of $\pi$ to be a sequence
\[ i, \pi (i), \pi^2 (i), \ldots , \pi^{k+1} (i) = j \]
where $1 \leq i,j \leq n$, and for all $0 < p < k$, $\pi^p (i) > n$. We write $\Ppi (i)$ for the $n$-path starting from $i$, which is clearly unique if it exists, and also $\ppi (i) = j$. We write $\Ppio (i)$ for the set of elements of $\{ n+1 , \dots , n+m \}$ appearing in the sequence.
A \emph{loop} of $\pi$ is defined to be a cycle
\[ j, \pi (j), \ldots , \pi^{k+1}(j) = j \]
where $n < j \leq n+m$. We write $\LL (\pi )$ for the set of all loops of $\pi$.

\begin{proposition}
\label{permprop}
The following holds for any permutation $\pi \in S(n+m)$:
\begin{enumerate}
\item For each $i$, $1 \leq i \leq n$, $\Ppi (i)$ is well-defined.
\item $\ppi \in S(n)$.
\item The family of sets
\[ \{ \Ppio (i) \mid 1 \leq i \leq n \} \cup  \LL (\pi ) \]
form a partition of $\{ n+1 , \ldots , n+m\}$.
\end{enumerate}
\end{proposition}
\begin{proof}
\begin{enumerate}
\item Consider the sequence
\[ i, \pi (i), \pi^2 (i) , \ldots \]
Either we reach $\pi ^{k+1} = j \leq n$, or there must be a least $l$ such that
\[ \pi^{k+1} (i) = \pi^{l+1} (i) > n, \qquad 0 \leq k < l . \]
(Note that the fact that $i \leq n$ allows us to write the left hand term as $\pi^{k+1} (i)$). But then, applying $\pi^{-1}$, we conclude that $\pi^{k} (i) = \pi^{l} (i)$, a contradiction.

\item If $\ppi (i) = \ppi (j)$, then $\pi^{k+1} (i) = \pi^{l+1} (j)$, where say $k \leq l$. Applying $(\pi^{-1})^{k+1}$, we obtain $i = \pi^{l-k}(j) \leq n$, whence $l = k$ and $i = j$.

\item It is standard that distinct cycles are disjoint. We can reason similarly to part (2) to show that 
if $\Ppio (i)$ meets $\Ppio (j)$, then $i = j$.
Similar reasoning to (1) shows that $\Ppio (i) \cap L = \varnothing$, for $L \in \LL (\pi)$. Finally, iterating $\pi^{-1}$ on $j>n$ either forms a cycle, or reaches $i \leq n$; in the latter case, $j \in \Ppio (i)$.
\end{enumerate}
\end{proof}

We now give a more algebraic description of the permutation $\ppi$.
Firstly, we extend our notation by defining $[n{:}m] := \{ n+1 , \ldots , m \}$.
Now we can write $[n{+}m] = [n] \dunion [n{:}n{+}m]$, where $\dunion$ is disjoint union.
We can use this decomposition to express $\pi \in S(n{+}m)$  as the disjoint union of the following four maps:
\[ \begin{array}{ll}
\pi_{1,1} : [n] \longrightarrow [n]  & \qquad \pi_{1,2} : [n] \longrightarrow [n{:}n{+}m] \\
\pi_{2,1} : [n{:}n{+}m] \longrightarrow [n]  & \qquad \pi_{2,2} : [n{:}n{+}m] \longrightarrow [n{:}n{+}m] 
\end{array} \]
We can view these maps as \emph{binary relations} on $[n{+}m]$ (they are in fact injective partial functions), and use relational algebra (union $R \cup S$, relational composition $R ; S$ and reflexive transitive closure $R^{\ast}$) to express $\ppi$ in terms of the $\pi_{ij}$:
\[ \ppi \;\; = \;\; \pi_{1,1} \; \cup \; \pi_{1, 2} ; \pi_{2, 2}^{\ast} ; \pi_{2, 1} . \]
We can also characterize the elements of $\LL (\pi )$:
\[ j \in \bigcup \LL (\pi ) \quad \Longleftrightarrow \quad \langle j, j \rangle \in \pi_{2,2}^{\ast} \cap \id_{[n{:}n{+}m]} . \]

\subsubsection{Loops}
We follow Kelly and Laplaza \cite{KL80} in making the following basic definitions.
The \emph{loops} of a category $\CC$, written $\Loop{\CC}$, are the endomorphisms of $\CC$ quotiented by the following equivalence relation: a composition
\begin{diagram}
A_1 & \rTo^{f_1} & A_2 & \rTo^{f_2} & \cdots & A_k & \rTo^{f_k} & A_1 \\
\end{diagram}
is equated with all its cyclic permutations.
A \emph{trace function} on $\CC$ is a map on the endomorphisms of $\CC$ which respects this equivalence.
We note in particular  the following standard result \cite{JSV}:
\begin{proposition}
If $\CC$ is traced, then the trace applied to endomorphisms:
\[ g : A \rightarrow A \;\; \longmapsto \;\; \Tr_{A, A}^{\II} (f) : \II \rightarrow \II \]
is a (scalar-valued) trace function.
\end{proposition}

\subsubsection{Traces of decomposable morphisms}
We now turn to a general proposition about traced categories, from which the structure of the free category will be readily apparent. It shows that whenever a morphism is decomposable into a tensor product followed by a permutation (as all morphisms in $\FSM$ are), then the trace can be calculated explictly by composing over paths. 

\begin{proposition}
\label{traceprop}
Let $\CC$ be a traced symmetric monoidal category, and consider a morphism of the form
\[ f = \pi^{-1} \circ \bigotimes_{i=1}^{n+m} f_i :  C 
\longrightarrow D  \]
where $C = \bigotimes_{i=1}^{n} A_i \otimes \bigotimes_{j= n+1}^{n+m} U_j$, $D =   \bigotimes_{i=1}^{n} B_i \otimes \bigotimes_{j= n+1}^{n+m} U_j$,
$\pi \in S(n+m)$, and $f_i : C_i \rightarrow D_{\pi (i)}$. 
Then
\[ \Tr_{C,D}^U (f) = \left(\prod_{l \in \LL (\pi )} s_l \right) \sdot (\ppi^{-1} \circ \bigotimes_{i=1}^n g_n ) \]
where for each $1 \leq i \leq n$, with $n$-path 
\[ \Ppi (i) = i, p_1 , \ldots , p_k , j \]
$g_i$ is the composition
\begin{diagram}
A_i & \rTo^{f_i} & U_{p_1} & \rTo^{f_{p_1}} & \cdots \cdots & U_{p_k} & \rTo^{f_{p_k}} & B_j \\
\end{diagram}
and for $l = p_1 , \cdots , p_k , p_1 \in \LL (\pi )$,
$s_l = \Tr_{\II , \II}^{U_{p_1}}(f_{p_k} \circ \cdots \circ f_{p_1})$.
The product $\prod_{l} s_l$ refers to multiplication in the monoid of scalars, which we know by Proposition~\ref{scprop} to be commutative.
\end{proposition}

\noindent Taken together with the following instance of Superposing:
\begin{equation}
\label{sdottr}
\Tr_{A, B}^U (s \sdot f) = s \sdot \Tr_{A, B}^U (f)  
\end{equation}
this Proposition yields a closed form description of the trace on expressions of the form:
\begin{equation}
\label{trcform}
s \sdot (\pi^{-1} \circ \bigotimes_j f_j ) \; : \; \bigotimes_j A_j \; \longrightarrow \; \bigotimes_j B_j  .
\end{equation}
We approach the proof of this Proposition via a number of lemmas.

Firstly, a simple consequence of Feedback Dinaturality:
\begin{lemma}
\label{permfeed}
Let $U = \bigotimes_{i=1}^n U_i$, and $\sigma \in S(n)$.
Let $\sigma U = \bigotimes_{i=1}^n U_{\sigma (i)}$.
Then
\[ \Tr_{A,B}^U (f) = \Tr_{A,B}^{\sigma (U)} ((1_A \otimes\sigma ) \circ f \circ (1_B \otimes \sigma^{-1})) . \]
\end{lemma}

\begin{lemma}
\label{tenslemma}
\[ \Tr_{A, B}^U (f) \otimes \Tr_{C, D}^V (g) = \Tr_{A \otimes C, B \otimes D}^{V \otimes U} ( (1_A \otimes \tau_{U, D \otimes V}) \circ (f \otimes g) \circ (1_{A} \otimes \tau_{C \otimes V, U}) ) . \]
\end{lemma}

\noindent The proof is in the Appendix.

We now show how the trace is evaluated along cyclic paths of any length.
We write $\sigma_{k+1} = 
\begin{pmatrix}
1 & 2 & \cdots & k & k+1 \\
2 & 3 & \cdots & k+1 & 1 \end{pmatrix}$, the cyclic permutation of length $k+1$. Note the useful recursion formula:
\begin{equation}
\label{shiftdecomp}
\sigma_{k+1} = (\tau \otimes 1) \circ (1 \otimes \sigma_k ) .
\end{equation}
Suppose we have morphisms $f_i : A_i \rightarrow A_{i+1}$, $1 \leq i \leq k+1$.
We write $U = \bigotimes_{i=2}^{k+1} A_i$, and $V = \bigotimes_{i=3}^{k+1} A_i$. (By convention, a tensor product over an empty range of indices is taken to be the tensor unit $\II$).
\begin{lemma}
\label{pathlemma}
For all $k \geq 0$:
\[ \Tr_{A_{1}, A_{k+2}}^{U} (\sigma_{k+1} \circ \bigotimes_{i=1}^{k+1}f_i ) = f_{k+1} \circ f_k \circ \cdots \circ f_1  . \]
\end{lemma}

\noindent The proof is relegated to the Appendix.
This lemma simultaneously generalizes Vanishing I ($k=0$) and Yanking ($k=1$, $A_1 = A_2 = A_3$, $f_1 = f_2 = 1_{A_1}$), and also the Generalized Yanking of \cite{AHS}.
The geometry of the situation is made clear by the following diagram.
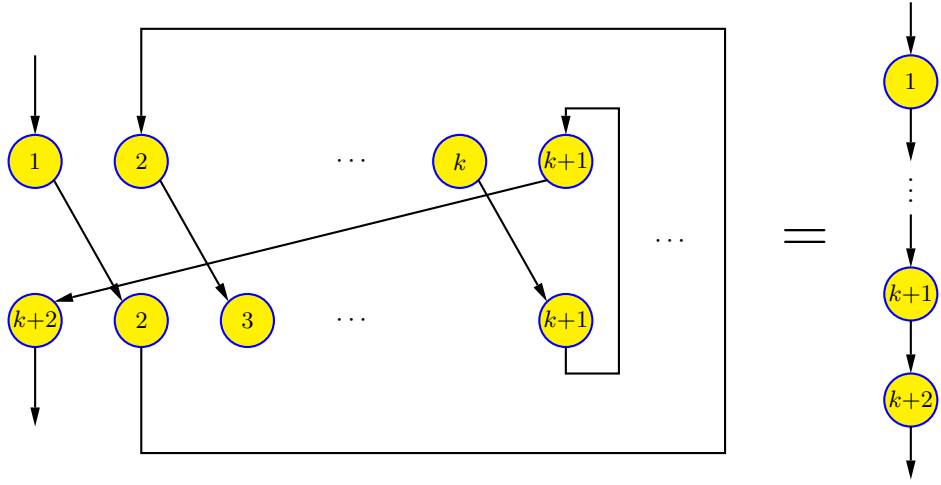
\begin{center}
\psset{unit=1in,cornersize=absolute}%
\begin{pspicture}(0,-1.666667)(4.861111,0.833333)
\pscircle[fillstyle=solid,fillcolor=yellow,linecolor=blue](0.138889,0){0.144424}
\rput(0.138889,0){{\small $1$}}
\pscircle[fillstyle=solid,fillcolor=yellow,linecolor=blue](0.694444,0){0.144424}
\rput(0.694444,0){{\small $2$}}
\pscircle[fillstyle=solid,fillcolor=yellow,linecolor=blue](2.361111,0){0.144424}
\rput(2.361111,0){{\small $k$}}
\pscircle[fillstyle=solid,fillcolor=yellow,linecolor=blue](2.916667,0){0.144424}
\rput(2.916667,0){{\small ${k{+}1}$}}
\pscircle[fillstyle=solid,fillcolor=yellow,linecolor=blue](0.138889,-0.833333){0.144424}
\rput(0.138889,-0.833333){{\small $k{+}2$}}
\pscircle[fillstyle=solid,fillcolor=yellow,linecolor=blue](0.694444,-0.833333){0.144424}
\rput(0.694444,-0.833333){{\small $2$}}
\pscircle[fillstyle=solid,fillcolor=yellow,linecolor=blue](1.25,-0.833333){0.144424}
\rput(1.25,-0.833333){{\small $3$}}
\pscircle[fillstyle=solid,fillcolor=yellow,linecolor=blue](2.916667,-0.833333){0.144424}
\rput(2.916667,-0.833333){{\small $k{+}1$}}
\psline[arrowsize=0.05in 0,arrowlength=2,arrowinset=0]{->}(0.237098,-0.098209)(0.596235,-0.735124)
\psline[arrowsize=0.05in 0,arrowlength=2,arrowinset=0]{->}(0.792654,-0.098209)(1.151791,-0.735124)
\psline[arrowsize=0.05in 0,arrowlength=2,arrowinset=0]{->}(2.45932,-0.098209)(2.818457,-0.735124)
\psline[arrowsize=0.05in 0,arrowlength=2,arrowinset=0]{->}(2.818457,-0.098209)(0.237098,-0.735124)
\psline[arrowsize=0.05in 0,arrowlength=2,arrowinset=0]{->}(2.916667,-0.972222)(2.916667,-1.111111)
(3.194444,-1.111111)
(3.194444,0.277778)
(2.916667,0.277778)
(2.916667,0.138889)
\psline[arrowsize=0.05in 0,arrowlength=2,arrowinset=0]{->}(0.694444,-0.972222)(0.694444,-1.527778)
(3.75,-1.527778)
(3.75,0.694444)
(0.694444,0.694444)
(0.694444,0.138889)
\rput(3.472222,-0.416667){$\cdots$}
\rput(1.805556,0){$\cdots$}
\rput(1.805556,-0.833333){$\cdots$}
\psline[arrowsize=0.05in 0,arrowlength=2,arrowinset=0]{<-}(0.138889,0.138889)(0.138889,0.555556)
\psline[arrowsize=0.05in 0,arrowlength=2,arrowinset=0]{->}(0.138889,-0.972222)(0.138889,-1.388889)
\rput(4.166667,-0.416667){{\Huge $=$}}
\pscircle[fillstyle=solid,fillcolor=yellow,linecolor=blue](4.722222,0.416667){0.144424}
\rput(4.722222,0.416667){{\small $1$}}
\pscircle[fillstyle=solid,fillcolor=yellow,linecolor=blue](4.722222,-0.694444){0.144424}
\rput(4.722222,-0.694444){{\small ${k{+}1}$}}
\pscircle[fillstyle=solid,fillcolor=yellow,linecolor=blue](4.722222,-1.25){0.144424}
\rput(4.722222,-1.25){{\small ${k{+}2}$}}
\psline[arrowsize=0.05in 0,arrowlength=2,arrowinset=0]{<-}(4.722222,0.555556)(4.722222,0.833333)
\psline[arrowsize=0.05in 0,arrowlength=2,arrowinset=0]{->}(4.722222,0.277778)(4.722222,0)
\rput(4.722222,-0.111111){$\vdots$}
\psline[arrowsize=0.05in 0,arrowlength=2,arrowinset=0]{<-}(4.722222,-0.555556)(4.722222,-0.277778)
\psline[arrowsize=0.05in 0,arrowlength=2,arrowinset=0]{->}(4.722222,-0.833333)(4.722222,-1.111111)
\psline[arrowsize=0.05in 0,arrowlength=2,arrowinset=0]{->}(4.722222,-1.388889)(4.722222,-1.666667)
\end{pspicture}%
\end{center}

\paragraph{Proof of Proposition~\ref{traceprop}}
Note that for $k = 0$, this is just Vanishing I.
Up to conjugation by some permutation $\sigma$, we can express $\pi$ as the tensor product of its $n$-paths and loops:
\[ \pi = \sigma \circ \left(\bigotimes_{i=1}^n \Ppi (i) \otimes \bigotimes_{L \in \Loop{\pi}} L\right) \circ \sigma^{-1} . \]
Using Lemmas~\ref{tenslemma} and~\ref{permfeed}, we can express the trace of $f$ in terms of the traces of the morphisms corresponding to the $n$-paths and loops of $\pi$. The trace of each $n$-path is given by Lemma~\ref{pathlemma}.
\qed

\subsubsection{Description of $\FTr$}

The objects are as for $\FSM$.
A morphism now has the form $(S, \pi , \lambda)$, where $(\pi , \lambda )$ are as in $\FSM$, and $S$ is a \emph{multiset of loops} in $\Loop{\CC}$, \ie an element of $\mathcal{M}(\Loop{\CC})$, the free commutative monoid generated by $\Loop{\CC}$.

Note that such a morphism 
\[ (S, \pi , \lambda) : (n, A) \longrightarrow (n, B) \]
can be written as
\begin{equation}
\label{trform} 
\left(\prod_{[s_i : A \rightarrow A]_{\sim} \in S} \Tr_{\II , \II}^{A}(s_i )\right) \sdot (\pi^{-1} \circ \bigotimes_{i=1}^n \lambda_i ) \; : \; \bigotimes_i A_i \longrightarrow \bigotimes_i B_i 
\end{equation}
in the language of traced symmetric monoidal categories. This will be our closed-form description of morphisms in the free traced category. It follows from Proposition~\ref{traceprop}, together with equations (\ref{sdotident})--(\ref{sdotten}), (\ref{smoncomp}), (\ref{smontens}), (\ref{sdottr}), that this is indeed closed under the traced monoidal operations.

\noindent We define the main operations on morphisms.

\subsubsection{Composition}
\[ (T, \sigma , \mu ) \circ (S, \pi , \lambda ) = (S \munion T, \sigma \circ \pi , i \mapsto \mu_{\pi (i)} \circ \lambda_i ) \]

\subsubsection{Tensor product}
\[ (S, \pi , \lambda ) \otimes (T, \sigma , \mu ) = (S \munion T, \pi \ptensor \sigma , [\lambda , \mu]) \]

\subsubsection{Trace}
\[ \Tr_{n,n}^m (S, \pi , \lambda ) = (S \munion T, \ppi, \mu ) \]
where
\[ \begin{array}{ll}
T = \mlb [\lambda_{\pi^l (j)} \circ \cdots \circ \lambda_j ]_{\sim} \mid \pi^{l+1}(j) = j \in \LL (\pi ) \mrb , \\
\mu : i \mapsto \lambda_{\ppi (i)} \circ \lambda_{\pi^k (i)} \circ \cdots \circ \lambda_i .  
\end{array}
\]

\noindent Note that $F_{\Tr}(\CC )(\II , \II ) = \mathcal{M}(\Loop{\CC})$. 
Also,
\[ F_{\Tr}(\mathbf{1}) = (\coprod_{n \in \Nat} S(n)) \times (\Nat , {+}, 0) . \]
That is, the objects in this free category are the natural numbers; a morphism is a pair $(\pi , n)$, where $\pi$ is a permutation, and $n$ is a natural number counting the number of loops. 

\subsection{Compact Closed Categories}
The free construction for compact closed categories was characterized in the pioneering paper by Kelly and Laplaza \cite{KL80}. Their construction is rather complex. Even when simplified to the strict monoidal case, several aspects of the construction are bundled in together, and it can be hard to spot what is going one. (For example, the path construction we gave for the trace in the previous section is implicit in their paper --- but not easy to spot!). We are now in a good position to disentangle and clarify their construction. Indeed, we have already explictly constructed $F_{\Tr}(\CC )$, and there is the $\GG$ or Int construction of Joyal, Street and Verity \cite{JSV}\footnote{Prefigured in \cite{AJ92}, and also in some unpublished lectures of Martin Hyland \cite{Hypcomm}.}, which is developed in the symmetric monoidal context with connections to Computer Science issues and the Geometry of Interaction in \cite{Abr96}. This construction gives the free compact closed category generated by a traced monoidal category. Thus we can recover the Kelly-Laplaza construction as the composition of these two adjunctions:
\[\begin{diagram}
\mathbf{Cat} && \pile{\rTo^{F_{\Tr}}\\ \bot \\ \lTo_{U_{\Tr}}} && \Tr{-} \mathbf{Cat}  && \pile{\rTo^{\GG}\\ \bot \\ \lTo_{U}} && \CCl{-} \mathbf{Cat} \\
\end{diagram}
\]
Adjoints compose, so
$F_{\CCl}(\CC ) = \mathcal{G} \circ F_{\Tr}(\CC )$.
This factorization allows us to `rationally reconstruct' the Kelly-Laplaza construction.

The main notion which has to be added to those already present in $F_{\Tr}(\CC )$ is that of \emph{polarity}. The ability to distincguish between positive and negative occurrences of a variable will allow us to transpose variables from inputs to outputs, or vice versa. This possibility of transposing variables means that we no longer have the simple situation that morphisms must be between lists of generating objects of the same length. However, note that in a compact closed category, $(A \otimes B)^* \simeq A^* \otimes B^*$, so any object constructed from generating objects by tensor product and duality will be isomorphic to one of the form
\[ \bigotimes_i A_i \otimes \bigotimes_j B^*_j . \]
Moreover, any morphism
\begin{equation}
\label{ccform}
f \; : \; \bigotimes_i A_i \otimes \bigotimes_j B^*_j \; \longrightarrow \;
\bigotimes_k C_k \otimes \bigotimes_l D^*_l
\end{equation}
will, after transposing the negative objects, be in biunique correspondence with one of the form
\begin{equation}
\label{transform}
\bigotimes_i A_i \otimes \bigotimes_l D_l \;  \longrightarrow \;
\bigotimes_k C_k \otimes \bigotimes_j B_j .
\end{equation}
A key observation is that \emph{in the free category, this transposed map~(\ref{transform}) will again be of the closed form~(\ref{trform})} which characterizes morphisms in $F_{\Tr}(\CC )$, as we saw in the previous section. From this, the construction of $F_{\CCl}(\CC )$ will follow directly.

\subsubsection{Objects}
The objects in $\FCC$ are, following the $\GG$ construction applied to $\FTr$, pairs of objects of $\FTr$, hence of the form $(n, m, A^{+}, A^{-})$, where
\[ A^{+} : [n] \longrightarrow \Ob{\CC} \qquad A^{-} : [m]  \longrightarrow \Ob{\CC} . \]
Such an object can be read as the tensor product
\[ \bigotimes_{i=1}^n A^{+}_i \otimes \bigotimes_{j=1}^m (A^{-}_{j})^* . \]
This is equivalent to the Kelly-Laplaza notion of signed set, under which objects have the form $(n, A, \sgn )$, where $\sgn : [n] \rightarrow \{ {+}, {-}\}$.

\subsubsection{Operations on objects}
The tensor product is defined componentwise on the positive and negative components. Formally: 
\[ (n,m, A^{+}, A^{-}) \otimes (p, q, B^{+}, B^{-}) = (n+p, m+q, [A^{+},B^{+}], [A^{-}, B^{-}]) . \]
The duality simply interchanges positive and negative components:
\[ (n, m, A^{+}, A^{-})^* = (m, n, A^{-}, A^{+}) . \]
Note that the duality is involutive, and distributes through tensor:
\[ A^{**} = A, \qquad (A \otimes B)^* = A^* \otimes B^* . \]

\subsubsection{Morphisms}
A morphism has the form
\[ (S, \pi , \lambda ) : (n, m, A^{+}, A^{-}) \longrightarrow (p, q, B^{+}, B^{-}) \]
where we require $n+q = k = m+p$, $\pi \in S(k)$, and $\lambda : [k] \rightarrow \Mor{\CC}$, such that
\[ \lambda_i : [A^{+}, B^{-}]_i \longrightarrow [B^{+}, A^{-}]_{\pi (i)} .  \]
$S$ is a multiset of loops, just as in $\FTr$. Note that $(S, \lambda , \pi )$ can indeed be seen as a morphism in $\FTr$ in the transposed form (\ref{transform}), as discussed previously.

We now describe the compact closed operations on morphisms.

\subsubsection{Composition}
Composition of a morphism $f : A \rightarrow B$ with a morphism $g : B \rightarrow B$ is given by feeding `outputs' by $f$ from the positive component of $B$ as inputs to $g$ (since for $g$, $B$ occurs negatively, and hence the positive and negative components are interchanged); and symmetrically, feeding the $g$ outputs from the negative components of $B$ as inputs to $f$.
This symmetry allows the strong form of duality present in compact closed categories to be interpreted in a very direct and natural fashion.

This general prescription is elegantly captured algebraically in terms of the trace, which co-operates with the duality to allow symmetric interaction between the two morphisms which are being composed. This is illustrated by the following diagram, which first appeared in \cite{AJ92}:

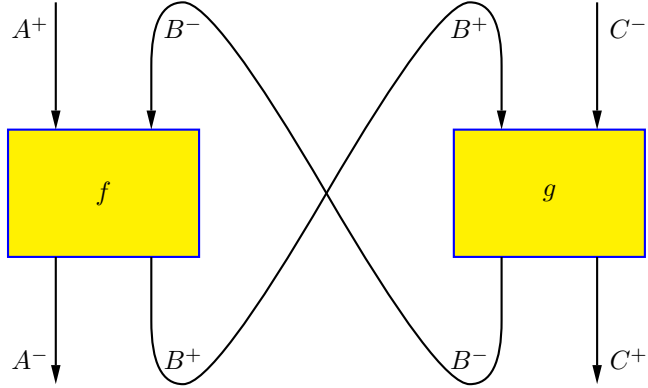
\begin{center}
\psset{unit=1in,cornersize=absolute,dimen=middle}%
\begin{pspicture}(0,-1)(3.333333,1)
\psframe[fillstyle=solid,fillcolor=yellow,linecolor=blue](0,-0.333333)(1,0.333333)
\rput(0.5,0){$f$}
\psframe[fillstyle=solid,fillcolor=yellow,linecolor=blue](2.333333,-0.333333)(3.333333,0.333333)
\rput(2.833333,0){$g$}
\psline[arrowsize=0.05in 0,arrowlength=2,arrowinset=0]{<-}(0.25,0.333333)(0.25,1)
\psline[arrowsize=0.05in 0,arrowlength=2,arrowinset=0]{->}(0.25,-0.333333)(0.25,-1)
\psline[arrowsize=0.05in 0,arrowlength=2,arrowinset=0]{<-}(3.083333,0.333333)(3.083333,1)
\psline[arrowsize=0.05in 0,arrowlength=2,arrowinset=0]{->}(3.083333,-0.333333)(3.083333,-1)
\psbezier[arrowsize=0.05in 0,arrowlength=2,arrowinset=0]{->}(0.75,-0.333333)(0.75,-0.444444)(0.75,-0.555556)(0.75,-0.666667)
(0.75,-0.888889)(0.805556,-1)(0.916667,-1)
(1.027778,-1)(1.277778,-0.666667)(1.666667,0)
(2.055556,0.666667)(2.305556,1)(2.416667,1)
(2.527778,1)(2.583333,0.888889)(2.583333,0.666667)
(2.583333,0.555556)(2.583333,0.444444)(2.583333,0.333333)
\psbezier[arrowsize=0.05in 0,arrowlength=2,arrowinset=0]{->}(2.583333,-0.333333)(2.583333,-0.444444)(2.583333,-0.555556)(2.583333,-0.666667)
(2.583333,-0.888889)(2.527778,-1)(2.416667,-1)
(2.305556,-1)(2.055556,-0.666667)(1.666667,0)
(1.277778,0.666667)(1.027778,1)(0.916667,1)
(0.805556,1)(0.75,0.888889)(0.75,0.666667)
(0.75,0.555556)(0.75,0.444444)(0.75,0.333333)
\rput(0.116667,0.866667){$A^+$}
\rput(0.116667,-0.866667){$A^-$}
\rput(3.25,0.866667){$C^-$}
\rput(3.25,-0.866667){$C^+$}
\rput(0.916667,0.866667){$B^-$}
\rput(0.916667,-0.866667){$B^+$}
\rput(2.416667,0.866667){$B^+$}
\rput(2.416667,-0.866667){$B^-$}
\end{pspicture}%
\end{center}

A concrete account for $\FCC$ follows directly from our description of the trace in $\FTr$:
chase paths, and compose (in $\CC$) the morphisms labelling the paths  to get the labels. In general, loops will be formed, and must be added to the multiset.
Formally, given arrows
\[ f : A \longrightarrow B, \qquad \qquad g : B \longrightarrow C \]
where
\[ \begin{array}{lcl}
f & = & (S, \pi , \lambda ) : (n, m, A^{+}, A^{-}) \longrightarrow (p, q, B^{+}, B^{-}) \\ 
g & = & (T, \sigma, \mu) : (p, q, B^{+}, B^{-}) \longrightarrow (r, s, C^{+}, C^{-}) 
\end{array}
\]
we form the composition
\[ (S \munion T \munion U, \; \EX{\pi}{\sigma} ,\; \rho ) : (n, m, A^{+}, A^{-}) \longrightarrow (r, s, C^{+}, C^{-}) . \]
There is an algebraic description of the permutation component $\EX{\pi}{\sigma}$, which can be derived from our algebraic description of $\ppi$ in the previous section. 
In the same manner as we did there, we can decompose each of $\pi$ and $\sigma$ into four components, which we write as matrices:
\[ \pi = \begin{pmatrix} \pi_{A^{+}A^{-}} \quad & \pi_{A^{+}B^{+}} \\
\pi_{B^{-}A^{-}} \quad & \pi_{B^{-}B^{+}}
\end{pmatrix}
\qquad \qquad
\sigma = \begin{pmatrix} \sigma_{A^{+}A^{-}} \quad & \sigma_{A^{+}B^{+}} \\
\sigma_{B^{-}A^{-}} \quad & \sigma_{B^{-}B^{+}}
\end{pmatrix}
\]
Now if we write 
\[ \EX{\pi}{\sigma} = \theta = \begin{pmatrix} \theta_{A^{+}A^{-}} \quad & \theta_{A^{+}C^{+}} \\
\theta_{C^{-}A^{-}} \quad &     \theta_{C^{-}C^{+}}
\end{pmatrix}
\]
then we can define
\[ \begin{array}{lcl}
\theta_{A^{+}A^{-}} & \;\; = \;\; & 
\pi_{A^{+}A^{-}} \; \cup \;\; \pi_{A^{+}B^{+}} ; \sigma_{B^{+}B^{-}} ; (\pi_{B^{-}B^{+}} ; \sigma_{B^{+}B^{-}})^{\ast} ; \pi_{B^{-}A^{-}} \\
\theta_{A^{+}C^{+}} & \;\; = \;\; & \pi_{A^{+}B^{+}}; (\sigma_{B^{+}B^{-}} ; \pi_{B^{-}B^{+}})^{\ast} ; \sigma_{B^{+}C^{+}} \\
\theta_{C^{-}A^{-}} & \;\; = \;\; & \sigma_{C^{-}B^{-}}; (\pi_{B^{-}B^{+}} ; \sigma_{B^{+}B^{-}})^{\ast} ; \pi_{B^{-}A^{-}} \\
\theta_{C^{-}C^{+}} & \;\; = \;\; & \sigma_{C^{-}C^{+}} \; \cup \;\; \sigma_{C^{-}B^{-}} ; \pi_{B^{-}B^{+}} ; (\sigma_{B^{+}B^{-}} ; \pi_{B^{-}B^{+}})^{\ast} ; \sigma_{B^{+}C^{+}} .
\end{array}
\]
This is essentially the 
`Execution formula' \cite{Gir89} --- see also \cite{JSV} and \cite{Abr96}; it appears implicitly in \cite{KL80} as a coequaliser.

Similarly, we can characterize the loops formed by composing $\pi$ and $\sigma$, $\LL (\pi , \sigma )$, by
\[ \begin{array}{ccl}
j & \;\in\; & \bigcup \LL (\pi , \sigma ) \;\;  \Longleftrightarrow \\ 
\langle j, j \rangle & \;\in\; & ((\pi_{B^{-}B^{+}} ; \sigma_{B^{+}B^{-}})^{\ast} \, \cap \, \id_{B^{-}}) \; \cup \;
((\sigma_{B^{+}B^{-}} ; \pi_{B^{-}B^{+}})^{\ast} \, \cap \, \id_{B^{+}}) . 
\end{array}
\]

The labelling function $\rho$ simply labels $\EX{\pi}{\sigma} : i \mapsto j$ with the arrow in $\CC$ formed by composing the arrows labelling the arcs in the path from $i$ to $j$ described by the above `flow matrix'. Similarly, $U$ is the multiset of loops labelling the cycles in $\LL (\pi , \sigma )$.

One can give algebraic descriptions of $\rho$ and $U$ by reformulating $\lambda$ and $\mu$ as \emph{graph homomorphisms} into (the underlying graph of) $\CC$. One can then form a homomorphism $\nu : G \rightarrow U_{\Grph}\CC$ from a combined graph $G$, which gives an `intensional description' of the composition. This combined graph will comprise the disjoint union of the graphs corresponding to the two arrows being composed, together with explicit \emph{feedback arcs}, labelled by $\nu$ with identity arrows in $\CC$.
One then considers the \emph{path category} $G^*$ freely generated from this graph \cite[2.VII]{Mac}; the above flow expressions for $\EX{\pi}{\sigma}$ yield a description of the paths in this category when relational composition is reinterpreted as concatenation of paths. We can then read off $\rho$ and $U$ from the unique functorial extension of  $\nu$ to this path category.

\subsubsection{Tensor Product}

This is defined componentwise as in $\FTr$, with appropriate permutation of indices in order to align positive and negative components correctly.

\subsubsection{Units and Counits}
Firstly, we describe the identity morphisms explicitly:
\[ \id = (\varnothing , \id_{[n+m]}, i \mapsto 1_{[A,B]_i}): (n, m, A^{+}, A^{-}) \longrightarrow (n, m, A^{+}, A^{-}) . \]
We join each dot in the input to the corresponding one in the output, and label it with the appropriate identity arrow.

Now consider the unit $\eta_A : \II \rightarrow A^* \otimes A$.
Once we have unpacked the definition of what comprises an arrow of this type, we see that we can make exactly the same definition as for the identity!
The unit is just \emph{the right transpose of the identity}. Similarly, the counit is \emph{the left transpose of the identity}.

Thus identities, units and counits are essentially all the same, except that the polarities allow variables to be transposed freely between the domain and codomain.

\[ \begin{array}{rc}
\text{Identity:} & \qquad \qquad
\begin{diagram}
\bullet^+ & \pile{\rTo^1 \\ \vdash} & \bullet^+ \\
\end{diagram}
\\
\text{Unit:} & \qquad \qquad
\begin{diagram}
\vdash &  \bullet^- & \rTo^1 & \bullet^+ \\
\end{diagram}
\\
\text{Counit:} & \qquad \qquad
\begin{diagram}
\bullet^+ & \rTo^1 & \bullet^- &  \vdash \\
\end{diagram}
\end{array}
\]

\subsection{Strongly Compact Closed Categories}
We now wish to analyze the new notion of strongly compact closed category in the same style as the previous constructions. Fortunately, there is a simple observation which makes this quite transparent. \emph{Provided that the category we begin with is already equipped with an involution} (but no other structure), then this involution `lifts' through all our constructions, yielding the free `dagger version' (in the sense of \cite{Sel05}) of each of our constructions. In particular, our construction of $\FCC$ in the previous section in fact gives rise to the free strongly compact closed category.

More precisely, we shall describe an adjunction
\begin{diagram}
\mathbf{InvCat} && \pile{\rTo^{F_SCC}\\ \bot \\ \lTo_{U_SCC}} && \mathbf{SCC{-}Cat} \\
\end{diagram}
where
$\mathbf{InvCat}$ is the category of categories with a specified \emph{involution}, \ie an identity on objects, contravariant, involutive functor; and functors preserving the involution.

Our previous construction of $\FCC$ lifts directly to this setting. The main point is that we can define an involution $()^{\dagger}$ on $\FCC$, under the assumption that we are given a primitive $()^{\dagger}$ on the generating category $\CC$. The dagger on $\FCC$ will endow it with the structure of a strongly compact closed category (for which the compact closed part will coincide with that already described for $\FCC$).

Given 
\[ (S , \pi , \lambda) : (n, m, A^{+}, A^{-}) \longrightarrow (p, q, B^{+}, B^{-}) , \]
we can define
\[ (S , \pi , \lambda)^{\dagger} = (\mlb [s^{\dagger}]_{\sim} \mid [s]_{\sim} \in S \mrb , \pi^{-1}, j \mapsto \lambda_{\pi^{-1}(j)}^{\dagger}) . \]
In short, we reverse direction on the arrows connecting the dots (including reversing the direction of loops), and label the reversed arrows with the reversals of the original labels.
This contrasts with the dual $f^*$, which by the way types are interpreted in this free situation, is essentially \emph{the same combinatorial object} as $f$, but with a different `marking' by polarities  --- there are no reversals involved.
Thus, if we had a labelling morphism
\[ \begin{diagram} \lambda_i = [A^{+}, B^{-}]_i & \rTo^{f_i} & [B^{+}, A^{-}]_{\pi (i) = j} 
\end{diagram}
\]
then we will get
\begin{diagram}
(\lambda^{\dagger})_j = [B^{+}, A^{-}]_j & \rTo^{f_i^{\dagger}} & [A^{+}, B^{-}]_{\pi^{-1}(j) = i} .
\end{diagram}
It is easy to see that $\eta_A = \epsilon_A^{\dagger}$, so this is compatible with our previous construction of $\FCC$.

\subsection{Parameterizing on the Monoid}
So far, the scalars have arisen intrinsically from the loops in the generating category $\CC$. However, we may wish for various reasons to be able to `glue in' a preferred multiplicative monoid of scalars into our traced, compact closed,or strongly compact closed category, For example, we may wish to consider only a few generating morphisms, but to take the complex numbers $\mathbb{C}$ as scalars. We will present a construction which accomodates this, as a simple refinement of the previous ones. There are versions of this construction for each of the traced, compact closed, and strongly compact closed cases: we shall only discuss the last of these.

Firstly, we note that there is a functor
\[ \LL : \InvCat \longrightarrow \InvSet \]
which sends a category to its sets of loops. The dagger defines an involution on the set of loops. Involution-preserving functors induce involution-preserving functions on the loops.

Now let $\InvCMon$ be the category of commutative monoids with involution, and involution-preserving homomorphisms. There is an evident forgetful functor
$U_{\InvCMon} \longrightarrow \InvSet$.
We can form the comma category $(\LL \downarrow U_{\InvCMon})$, whose objects are of the form $(\CC , \varphi , M)$, where $\varphi$ is an involution-preserving map from $\Loop{\CC}$ to the underlying set of $M$. Here we can think of $M$ as the prescribed monoid of scalars, and $\varphi$ as specifying how to evaluate loops from $\CC$ in this monoid.

\noindent There is a forgetful functor
$U_{\VV} : \SCC{-}\mathbf{Cat} \longrightarrow \VV$
\[ U_{\VV} : \CC \longmapsto (U_{\SCC} (\CC ), \; f : A \rightarrow A  \longmapsto \Tr_{\II , \II}^A (f), \; \CC (\II , \II )) . \]
Our task is to construct an adjunction
\begin{diagram}
\mathbf{\VV} && \pile{\rTo^{F_{\VV}}\\ \bot \\ \lTo_{U_{\VV}}} && \mathbf{SCC{-}Cat} \\
\end{diagram}
which builds the free SCC on a category \emph{with prescribed scalars}.
This is a simple variation on our previous construction of $\FSCC$,
which essentially acts by composition with the loop evaluation function $\varphi$ on $\FSCC$.
We use the prescribed monoid $M$ in place of $\MM (\Loop{\CC})$.
Thus a morphism in $\FV$ will have the form $(m, \pi , \lambda )$, where $m \in M$. Multiset union is replaced by the monoid operation of $M$. The action of the dagger functor on elements of $M$ is by the given involution on $M$. When loops in $\CC$ arise in forming compositions in the free category, they are evaluated in $M$ using the function $\varphi$.

The monoid of scalars in this free category will of course be $M$.

\newpage
\section*{Appendix}
The following equational proofs involve some long typed formulas. To aid in readability, we have annotated each equational step (reading down the page) by $\underline{\text{\emph{underlining}}}$ each redex, and $\overline{\text{\emph{overlining}}}$ the corresponding  contractum.

\subsubsection{Proof of Lemma \ref{tenslemma}}
\begin{proof}
\[
\begin{array}{ll}
& \underline{\Tr_{A, B}^U (f) \otimes \Tr_{C, D}^V (g)} \\
=  & \text{$\{$Superposing$\}$} \\
& \overline{\Tr_{A \otimes C, B \otimes D}^{V} ( \underline{\Tr_{A, B}^U (f) \otimes g})} \\
=  & \text{$\{$Naturality of $\tau$$\}$} \\
& \Tr_{A \otimes C, B \otimes D}^{V} (\overline{\tau_{D \otimes V, B} \circ (\underline{g \otimes  \Tr_{A, B}^U (f)}) \circ \tau_{A, C \otimes V}}) \\
=  & \text{$\{$Superposing$\}$} \\
& \Tr_{A \otimes C, B \otimes D}^{V} (\underline{\tau_{D \otimes V, B}} \circ \overline{\Tr_{C \otimes V \otimes A, D \otimes V \otimes B}^U (g \otimes  f)} \circ \underline{\tau_{A, C \otimes V}}) \\
=  & \text{$\{$Input/Output Naturality$\}$} \\
& \Tr_{A \otimes C, B \otimes D}^{V} ( \Tr_{A \otimes C \otimes V, B \otimes D \otimes V}^{U}(\underline{\overline{(\tau_{D \otimes V, B} \otimes 1_{U})}  \circ (g \otimes f) \circ \overline{(\tau_{A, C \otimes V} \otimes 1_U )}})) \\
=  & \text{$\{$$\SMon$ Coherence$\}$} \\
& \underline{\Tr_{A \otimes C, B \otimes D}^{V}} (\underline{\Tr_{A \otimes C \otimes V, B \otimes D \otimes V}^{U}}(\overline{(1_A \otimes \tau_{U, D \otimes V}) \circ (f \otimes g) \circ (1_{A} \otimes \tau_{C \otimes V, U})})) \\
= & \text{$\{$Vanishing II$\}$} \\
& \overline{\Tr_{A \otimes C, B \otimes D}^{V \otimes U}} ( (1_A \otimes \tau_{U, D \otimes V}) \circ (f \otimes g) \circ (1_{A} \otimes \tau_{C \otimes V, U})) . 
\end{array}
\]
\end{proof}

\newpage
\subsubsection{Proof of Lemma \ref{pathlemma}}
\begin{proof}
Note that for $k = 0$, this is just Vanishing I. We now reason inductively when $k > 0$.
\[
\begin{array}{ll}
& \underline{\Tr_{A_{1}, A_{k+2}}^U}  (\sigma_{k+1} \circ \bigotimes_{i=1}^{k+1}f_i ) \\
= & \text{$\{$Vanishing II$\}$} \\
& \overline{\Tr_{A_{1}, A_{k+2}}^{A_2}} ( \overline{\Tr_{A_{1} \otimes A_2 , A_{k+2}\otimes A_2 }^V} (\underline{\sigma_{k+1} \circ \bigotimes_{i=1}^{k+1}f_i} )) \\
= & \text{$\{$(\ref{shiftdecomp})$\}$} \\
& \Tr_{A_{1}, A_{k+2}}^{A_2} ( \Tr_{A_{1} \otimes A_2 , A_{k+2}\otimes A_2 }^V ( \overline{\underline{(\tau \otimes 1)} \circ (1 \otimes \sigma_k ) \circ \bigotimes_{i=1}^{k+1}f_i} )) \\
= &  \text{$\{$Naturality in $A_{k+2} \otimes A_2$$\}$}  \\
&  \Tr_{A_{1}, A_{k+2}}^{A_2} ( \overline{\tau} \circ \Tr_{A_{1} \otimes A_2 , A_{2}\otimes A_{k+2} }^V ( (1 \otimes \sigma_k ) \circ \underline{\bigotimes_{i=1}^{k+1}f_i} )) \\
= & \text{$\{$Bifunctoriality of $\otimes$$\}$}   \\
& \Tr_{A_{1}, A_{k+2}}^{A_2} ( \tau \circ \Tr_{A_{1} \otimes A_2 , A_{2}\otimes A_{k+2} }^V ( (1 \otimes \sigma_k ) \circ \overline{(1 \otimes \bigotimes_{i=2}^{k+1}f_i ) \circ (\underline{f_1 \otimes 1_{U}})})) \\
= & \text{$\{$Naturality in $A_{1} \otimes A_2$$\}$}  \\
& \Tr_{A_{1}, A_{k+2}}^{A_2} ( \tau \circ \Tr_{A_{2} \otimes A_2 , A_{2}\otimes A_{k+2} }^V ( (1 \otimes \sigma_k ) \circ (1 \otimes \bigotimes_{i=2}^{k+1}f_i )) \circ (\underline{\overline{f_1 \otimes 1_{A_2}}})) \\
= & \text{$\{$Naturality in $A_{1}$$\}$}  \\
& \Tr_{A_{2}, A_{k+2}}^{A_2} ( \tau \circ \Tr_{A_{2} \otimes A_2 , A_{2}\otimes A_{k+2} }^V ( \underline{(1 \otimes \sigma_k ) \circ (1 \otimes \bigotimes_{i=2}^{k+1}f_i )}) ) \circ \overline{f_1} \\
= & \text{$\{$Bifunctoriality of $\otimes$$\}$}  \\
& \Tr_{A_{2}, A_{k+2}}^{A_2} ( \tau \circ \underline{\Tr_{A_{2} \otimes A_2 , A_{2}\otimes A_{k+2} }^V ( \overline{1 \otimes (\sigma_k  \circ \bigotimes_{i=2}^{k+1}f_i )}}) ) \circ f_1 \\
= & \text{$\{$Superposing$\}$}  \\
& \Tr_{A_{2}, A_{k+2}}^{A_2} ( \tau \circ (\overline{1 \otimes \underline{\Tr_{A_2 , A_{k+2} }^V (  \sigma_k  \circ \bigotimes_{i=2}^{k+1}f_i )}})) \circ f_1
\\
= &  \text{$\{$Induction hypothesis$\}$} \\
& \Tr_{A_{2}, A_{k+2}}^{A_2} ( \underline{\tau \circ (1 \otimes (\overline{f_{k+1} \circ \cdots \circ f_2} ))} ) \circ f_1 \\
= & \text{$\{$Naturality of $\tau$$\}$}  \\
& \Tr_{A_{2}, A_{k+2}}^{A_2} ( \overline{(\underline{(f_{k+1} \circ \cdots \circ f_2 ) \otimes 1}) \circ \tau} ) \circ f_1 \\
= & \text{$\{$Naturality in $A_{k+2}$$\}$}  \\
& \overline{(f_{k+1} \circ \cdots \circ f_2 )} \circ \underline{\Tr_{A_{2}, A_{2}}^{A_2} (  \tau )} \circ f_1 \\
= & \text{$\{$Yanking$\}$}  \\
& (f_{k+1} \circ \cdots \circ f_2 ) \circ \overline{1_{A_{2}}} \circ f_1    \\
= &  f_{k+1} \circ \cdots \circ f_2 \circ f_1 . 
\end{array}
\]
\end{proof}

\end{document}